\documentclass[sigconf,nonacm]{acmart}

\copyrightyear{2021}
\setcopyright{none}

\usepackage{algorithm,algorithmicx}
\usepackage[noend]{algpseudocode}
\usepackage{booktabs}
\usepackage{pgfplots}
  \pgfplotsset{compat=1.16}
  \usepgfplotslibrary{groupplots}
\usepackage{varwidth}
\usepackage[labelfont=bf, textfont=md, font=small,aboveskip=0pt,belowskip=0pt]{caption}
\usepackage{enumitem}
\usepackage{my-macros}
\usepackage{multirow}
\usepackage{xspace}

\providecommand{\clone}{\overline{N}} 
\providecommand{\cord}{\mathcal{CO}} 
\providecommand{\nord}{\mathcal{NO}} 
\providecommand{\figfont}{\footnotesize} 

\makeatletter%
\if@ACM@review
  \providecommand{\ourscan}{ParIndexSCAN} %
\else
  \providecommand{\ourscan}{GBBSIndexSCAN} %
\fi
\makeatother%

\newenvironment{alglinebreaks}
  {\begin{varwidth}[t]{\linewidth}}
  {\end{varwidth}}
\providecommand{\algind}{\hskip\algorithmicindent} 
\makeatletter
\algnewcommand{\InlineComment}[1]{\Statex \hskip\ALG@thistlm \null\hfill \(\triangleright\) #1}
\makeatother
\providecommand{\cvar}[1]{\mathit{#1}}  
\providecommand{\Pfor}[1]{\For{#1} \textbf{in parallel}} 

\makeatletter
\newcommand*{\whp}{%
    \@ifnextchar{.}%
        {w.h.p}%
        {w.h.p.\@\xspace}%
}
\makeatother

\definecolor{PineGreen}{rgb}{0.0, 0.47, 0.44}
\definecolor{MidnightGreen}{rgb}{0.0, 0.29, 0.33}
\definecolor{TyrianPurple}{rgb}{0.4, 0.01, 0.24}

\usepackage[compact]{titlesec}
\titlespacing*{\section}{0pt}{3pt}{3pt}
\titlespacing*{\subsection}{0pt}{2pt}{2pt}

\usepackage{times,microtype}

\begin{document}
\fancyhead{}

\title{Parallel Index-Based Structural Graph Clustering and Its Approximation}
\titlenote{This document is a non-commercial version of the paper appearing in
the 2021 ACM SIGMOD International Conference on Management of Data.}

\author{Tom Tseng}
\affiliation{%
  \institution{MIT CSAIL}}
\email{tomtseng@csail.mit.edu}

\author{Laxman Dhulipala}
\affiliation{%
  \institution{MIT CSAIL}}
\email{laxman@mit.edu}

\author{Julian Shun}
\affiliation{%
  \institution{MIT CSAIL}}
\email{jshun@mit.edu}


\begin{abstract}
  SCAN (Structural Clustering Algorithm for Networks) is a well-studied,
widely used graph clustering algorithm. For large graphs, however,
sequential SCAN variants are prohibitively slow, and parallel SCAN
variants do not effectively share work among queries with different
SCAN parameter settings.  Since users of SCAN often explore many
parameter settings to find good clusterings, it is worthwhile to
precompute an index that speeds up queries.

This paper presents a practical and provably efficient parallel
index-based SCAN algorithm based on GS*-Index, a recent sequential
algorithm. Our parallel algorithm improves upon the asymptotic work of
the sequential algorithm by using integer sorting.  It is also highly
parallel, achieving logarithmic span (parallel time) for both index
construction and clustering queries.  Furthermore, we apply
locality-sensitive hashing (LSH) to design a novel approximate SCAN
algorithm and prove guarantees for its clustering behavior.

We present an experimental evaluation of our algorithms on
large real-world graphs. On a 48-core machine with two-way
hyper-threading, our parallel index construction achieves
50--151$\times$ speedup over the construction of GS*-Index. In fact,
even on a single thread, our index construction algorithm is faster
than GS*-Index.  Our parallel index query implementation achieves
5--32$\times$ speedup over GS*-Index queries across a range of SCAN
parameter values, and our implementation is always faster than ppSCAN,
a state-of-the-art parallel SCAN algorithm. Moreover, our experiments
show that applying LSH results in faster index construction
while maintaining good clustering quality.

\end{abstract}

\maketitle

\section{Introduction}\label{chapter:intro}

In data mining and unsupervised learning, clustering is a fundamental technique
that organizes data into meaningful
groups.
Because much real-world data can be represented as graphs,
there is significant practical and theoretical interest in
\emph{graph clustering}, in which the goal is to partition the
vertices of a graph into clusters such that ``similar'' vertices fall
into the same cluster~\cite{bansal2004correlation, demaine2006correlation, schaeffer2007graph, pan2015parallel, shun2016parallel, yin2017local, bateni2017affinity, azad2018hipmcl}.
In particular, a good clustering usually has many edges that
fall within clusters and few edges that connect different
clusters.
Graph clustering is a popular problem with a wide range of
applications, including social and biological network
analysis~\cite{girvan2002community}, load balancing in
distributed systems~\cite{aydin2019distributed}, image
segmentation~\cite{tolliver2006graph}, natural language
processing~\cite{biemann2006chinese}, and recommendation
systems~\cite{bellogin2012using}.

One well-known approach to graph clustering is \emph{structural clustering}, which
Xu et al.\ first introduced via the Structural Clustering Algorithm for Networks
(SCAN)~\cite{xu2007scan}.
Structural clustering exploits the idea that vertices whose neighbor sets
resemble each other are ``similar,'' a type of homophily that is often satisfied in practice.
The approach is unique in that it also finds
\emph{hub} vertices that connect different clusters, as well as \emph{outlier}
vertices that lack strong ties to any cluster. Researchers have used SCAN
to find meaningful clusters in biological
data~\cite{mete2008structural, martha2011constructing, liu2011translating,
ding2012atbionet} and web data~\cite{papadopoulos2009leveraging,
papadopoulos2010graph, papadopoulos2010image, lin2010hierarchical,
schinas2015visual, schinas2015multimodal}.

SCAN as Xu et al.\ originally described it suffers from two issues: (1)
the costliness of sequentially computing similarities among
all adjacent vertices, and (2) the costliness of tuning the parameters of the
algorithm to achieve good clustering quality. Many researchers have developed
variants of SCAN to address these issues. To alleviate issue (1),
some variants exploit parallelism~\cite{chen2013pscan,
zhao2013pscan, stovall2014gpuscan, zhou2015sparkscan, takahashi2017scan,
che2018parallelizing, mai2018scalable} or introduce
algorithmic optimizations like pruning unnecessary similarity
computations~\cite{shiokawa2015scan++, chang2017pscan, che2018parallelizing}. To
alleviate issue (2), some variants precompute an index from which
computing the clusterings for different parameter values is
fast~\cite{bortner2010progressive,huang2012revealing,wen2017efficient}.
To be efficient on large graphs, SCAN-based algorithms should address
both issues, which existing algorithms fail to do.

This paper addresses the aforementioned issues by presenting a new parallel
index-based SCAN algorithm based on the sequential GS*-Index SCAN
algorithm~\cite{wen2017efficient}. Our algorithm achieves the same work bounds
as GS*-Index and is highly parallel, achieving logarithmic span (parallel time)
with high probability (\whp).\footnote{The \emph{work} of an
algorithm is the number of operations it performs. The \emph{span}
(parallel time) of an algorithm is the length of its longest
sequential dependence. We use \emph{with high probability} (\whp) to describe events
    that occur with probability at least $1 - 1/n^c$ where $n$ is the input size
and $c$ is some positive real number.} The key ingredients to achieve
our strong time bounds are the careful use of doubling search, as well as
parallel algorithms for graph connectivity and hash tables.
We also show how using matrix multiplication on dense graphs and using integer sort
improve the index construction work bound compared to GS*-Index's bound.

To further improve performance, we present new approximate algorithms
that use locality-sensitive hashing (LSH) to speed up similarity computation.
We provide a non-trivial theoretical analysis of the accuracy of our algorithms.
Our experiments show that using LSH speeds up index construction without
sacrificing clustering quality. We summarize the asymptotic running time bounds
for index construction in \cref{tab:running-times}.

We present optimized implementations of our algorithms.
The most important optimizations are a merge-based parallel triangle
counting algorithm described by Shun and Tangwongsan~\cite{shun2015multicore}
to compute similarities; concurrent union-find to compute connectivity for queries; and, for our approximate algorithms,
a heuristic to avoid using LSH
on low-degree vertices that would not benefit from approximation.
In our
experiments, our index construction algorithm achieves 50--151$\times$
speedup over the construction of GS*-Index for several large
real-world graphs on a machine with 48 cores and two-way
hyper-threading. In fact, our index construction algorithm is faster
than GS*-Index even when we run our algorithm on a
single thread. Furthermore, our parallel index query implementation,
which extracts a clustering for a specific set of parameters
from the index, achieves 5--32$\times$ speedup over
GS*-Index queries across a range of SCAN parameter values. Our implementation
also achieves faster query times on all tested parameter values
compared to ppSCAN~\cite{che2018parallelizing}, a state-of-the-art
parallel SCAN algorithm.

\begin{table}
  \vspace{-1pt}
  {\figfont
  \begin{tabular}{lll}
  \toprule
  {\bf Description} & {\bf Work} & {\bf Span} \\
  \midrule
  Exact index, weighted graph & $O((\alpha + \log n)m)$ \whp & $O(\log n)$ \whp \\
  \midrule[0.1pt]
  \multirow{2}{*}{Exact index, unweighted graph}
    & $O((\alpha + \log \log n) m )$ \whp & $O(\log n)$ \whp \\
    & $O(\alpha m)$ \whp & $O(n^\beta)$ \whp \\
  \midrule[0.1pt]
  \multirow{2}{*}{Approximate index}
    & $O((k + \log \log n) m )$ \whp & $O(\log n)$ \whp \\
    & $O(k m)$ & $O(n^\beta)$ \\
  \bottomrule
  \end{tabular}
  }
    \caption{Summary of asymptotic running time bounds for index construction.
    The arboricity of the input graph is $\alpha$, the number of samples
    used for approximation is $k$, and $0<\beta\leq 1$.
    For the exact indices on dense graphs, the $\alpha m$ work term may be replaced
    by $n^{\omega_p}$, where $n^{\omega_p} \le n^{2.373}$ is the asymptotic work to multiply
    two $n$-by-$n$ matrices in logarithmic span.
  }
  \label{tab:running-times}
\end{table}

The contributions of this paper are as follows:
\begin{enumerate}[label=(\textbf{\arabic*}),topsep=1pt,itemsep=0pt,parsep=0pt,leftmargin=15pt]

  \item We present a new parallel index-based SCAN algorithm that matches the work bounds of the sequential GS*-Index algorithm and
     has logarithmic span \whp. We also show how matrix multiplication and
     integer sorting improve the work bounds.
  \item We introduce the use of locality-sensitive hashing as an approximation
    technique for SCAN that is
    provably efficient and has behavior guarantees
    relative to exact SCAN.
  \item We implement and evaluate our algorithm on large real-world graphs. Our experiments demonstrate that our implementation
  outperforms other existing SCAN algorithms and confirm that
  locality-sensitive hashing provides running time
  improvements.
  \item
    \makeatletter%
    \if@ACM@review
      We release the implementation of our algorithm.\footnote{Anonymized code
        which will be available within two weeks of paper submission:
        \url{https://anonymous.4open.science/r/99415ef1-ecf2-4d62-945b-9c4a7a3230bc/}.
        We will provide a non-anonymous GitHub link for the final version
      of this paper.}
    \else
      We release the implementation of our
      algorithm.\footnote{Code: \url{https://github.com/ParAlg/gbbs/tree/master/benchmarks/SCAN/IndexBased}}
    \fi
    \makeatother%
\end{enumerate}

\section{Preliminaries}\label{chapter:prelim}

This section provides background definitions and concepts that
subsequent sections use.

\subsection{Set similarity}

\subsubsection{Similarity measures}\label{sec:sim-measures}

Two common measures for the similarity of two sets $A$ and $B$ with elements
from a finite universe $U$ are the Jaccard
similarity and the cosine similarity:
\begin{align*}
  \on{JaccardSim}(A, B) = \frac{\abs{A \cap B}}{\abs{A \cup B}}, \quad
  \on{CosineSim}(A, B) = \frac{\abs{A \cap B}}{\sqrt{\abs{A}}\sqrt{\abs{B}}}
.\end{align*}
If the sets are weighted and have
weight functions $w_A, w_B : U \mapsto \R$, then
there is a weighted form of cosine similarity:
\[
  \on{WeightedCosineSim}(A, B) =
  \frac{\sum_{x \in A \cap B}w_A(x)w_B(x)}{\sqrt{\sum_{x \in
  A}w_A(x)^2}\sqrt{\sum_{x \in B}w_B(x)^2}}
.\]
(There is also a weighted version of Jaccard similarity, which we do not consider
in this work.)

The cosine similarity is really a similarity measure between non-zero vectors;
given vectors $u$ and $v$ with an angle of $\theta$ between the two vectors, the
cosine similarity is defined as
\begin{align*}
  \on{CosineSim}(u, v) = \cos(\theta) = \frac{u \cdot v}{\norm{u}\norm{v}}
.\end{align*}
The definition of cosine similarity for sets with elements from $U$
follows by representing sets as vectors in $\R^U$ (namely, as a bit
vector for unweighted sets and as a vector of weights for weighted
sets).

\subsubsection{Locality-sensitive hashing}\label{sec:lsh}

Suppose that there is a collection of large sets with elements from a finite universe
$U$.
\emph{Locality-sensitive hashing} (LSH) is a technique to quickly approximate the
similarity between pairs of these sets. The idea is to devise a hash function
family that maps similar sets to similar, smaller
\emph{sketches}. We estimate similarities by precomputing all sketches
and operating on the sketches rather than on the large original sets.

A well-known LSH scheme for estimating Jaccard similarity,
for instance, is MinHash~\cite{broder1997resemblance}.  MinHash works by drawing
a uniformly random permutation $\pi$ on $U$ and considering the sketch of a non-empty
set $S$ to be $\min_{x \in S} \pi(x)$. For any two non-empty sets $A$ and $B$, the probability
that the sketches of $A$ and $B$ are equal is $\on{JaccardSim}(A, B)$.
To increase the precision at the cost of extra work, we fix a number of samples $k \in \N$ and perform this
process $k$ times independently to get $k$-length sketches.
The proportion of matching coordinates between two sketches is an estimate
of the Jaccard similarity between the two corresponding sets.
There are variants of
MinHash that are more computationally efficient such as
$k$-partition MinHash~\cite{li2012one}. There are also variants
for weighted Jaccard similarity~\cite{wu2020review}. Since the weighted variants are
more complicated and less practical, we do not use weighted
Jaccard similarity in this work.

SimHash~\cite{charikar2002similarity} is a well-known LSH
scheme for estimating the angle between two vectors. The idea
behind SimHash is to consider drawing a vector $v$ in $\R^U$ with uniformly
random direction by drawing each
coordinate independently from the standard normal distribution.
We take the sketch of a vector $u$ to be
$\on{sign}(u \cdot v)$. For a pair of non-zero vectors $a$ and $b$ with angle $\theta \in [0, \pi]$
in radians between them, the probability that the sketches of $a$ and $b$
differ is exactly $\theta / \pi$; because $v$ has uniformly random direction,
the orthogonal hyperplane to $v$ separates $a$ and $b$ with probability $\theta
/ \pi$, which exactly corresponds to the event that $\on{sign}(a \cdot v) \not=
\on{sign}(b \cdot v)$. Like with MinHash, to tune the precision, we repeat this
process $k \in \N$ times to get $k$-length sketches.
The number of differing
entries between the sketches of $a$ and $b$ multiplied by $\pi/k$ is an estimate
  $\hat{\theta} \sim
\on{Binomial}(k, \theta/\pi) \cdot \pi/k$,
  which in turn provides an estimate
  $\cos(\hat{\theta})$ for $\cos(\theta) = \on{CosineSim}(a, b)$.

\subsection{Graphs and graph notation}

We denote an unweighted, undirected graph $G$ by $G = (V, E)$, where $V$ is the
set of vertices and $E \subseteq \set{\set{u, v} : u, v \in V}$ is the set
of edges. We denote a weighted graph $G$ by $G = (V, E, w)$, where the weight
function $w : E
\to \R$ maps edges to weights. Following common
convention, we use $n$ to denote the number of vertices $\abs{V}$ and $m$
to denote the number of edges $\abs{E}$. The neighborhood $N(v)$ of a vertex
$v$ is the set of vertices connected to $v$ by an edge.  The \emph{closed
neighborhood} of $v$ is $\clone(v) = N(v) \cup \set{v}$. The degree of a vertex
is the size of its neighborhood, $N(v)$.

For directed graphs, each edge in $E$ becomes an ordered pair rather
than an unordered pair. The out-neighborhood of a vertex $v$ is the
set of all vertices $u$ such that $(v,u) \in E$.

The \emph{arboricity} $\alpha$ of an undirected graph $G$ is the minimum
number of spanning forests that covers all edges of the graph. The
arboricity is bounded below by $\ceil{m/(n-1)}$ since each spanning
forest covers at most $n-1$ edges and is bounded above by
$O\left(\sqrt{m + n}\right)$. A $\emph{triangle}$ is a triplet of
edges $\set{u,v}, \set{v,x}, \set{x, u}$ between distinct vertices $u,
v, x$ in $V$. There are triangle counting algorithms that find all
triangles in a graph in $O(\alpha m)$ time~\cite{chiba1985arboricity}.

We represent graphs as adjacency lists, in which each vertex
has a list of its neighbors.
We only consider simple graphs, i.e., graphs with at most one
edge between any pair of vertices and no self-loop
edges. We index vertices with the integers in the range $[1, n]$.

\subsection{Parallelism}

\subsubsection{Parallel programming model}

We design our algorithms for multicore shared-memory machines.
Readily available shared-memory machines are able to operate on
the largest publicly available real-world graphs, which have hundreds of billions of edges~\cite{dhulipala2018theoretically}.
Shared-memory systems are fast due to low communication costs and are
easier to program for than distributed systems are.

We use a fork-join programming model with arbitrary forking; a process can ``fork'' into an
arbitrary number of parallel processes in unit time and can ``join'' to synchronize among
forked processes. Most notably, a fork and a join suffice to implement a
parallel for-loop. We further assume that processes can concurrently read,
write, atomically add, and compare-and-swap at memory locations.
\emph{Compare-and-swap} (CAS) takes three arguments: a memory
location \emph{x}, an old value \emph{old\_V}, and a new value
\emph{new\_V}. If the value stored at \emph{x} is equal to \emph{old\_V},
the CAS atomically updates the value at \emph{x} to be \emph{new\_V} and
returns \emph{true}. Otherwise, the CAS returns \emph{false}. Almost all modern
processors support CAS.

We analyze  the complexity of algorithms with the work-span
model, a standard model for analyzing shared-memory parallel algorithms~\cite{jaja1992introduction, cormen2009introduction}. The
\emph{work} of a
program execution
is the total number of instructions executed, and the \emph{span}
is the length of the longest sequential critical path of instructions.
For a program with $W$ work and $S$ span, a work-stealing scheduler,
such as the one in Cilk~\cite{blumofe1999scheduling}, can execute the program in $W/P + O(S)$
expected time with $P$ processors. A parallel algorithm whose work
asymptotically matches the work of the most efficient known sequential algorithm
is \emph{work-efficient}, which is an important characteristic since $W/P$ is
often much higher than $S$ in well-designed parallel algorithms when run on large data sets.

\subsubsection{Parallel primitives}\label{sec:par-primitives}

This paper makes use of many existing parallel algorithms, which we
describe below.

\emph{Hash tables:}
Gil et al.\ present a hash table which supports inserting $k$
elements in $O(k)$ work and $O(\log^* k)$ span \whp.
Looking up an element takes $O(1)$
work~\cite{gil1991towards}.

\emph{Primitives on arrays:} The \emph{reduce} operation computes the
    sum of all elements in an array.
    (The sum operation is often the numerical addition operation but more
    generally may be any associative binary operation. For instance,
    \emph{reduce} can compute the maximum element in an array.)
    The \emph{filter} operation returns a subsequence of the original
    sequence consisting of all elements matching a user-specified predicate.
    For an array of $n$ elements, both operations run in $O(n)$ work and
    $O(\log n)$ span~\cite{jaja1992introduction, blelloch2010parallel}.
    The \emph{remove duplicates} operation returns an array that has the same
    set of elements as the original input array has, but without any duplicate
    elements. Removing duplicates using a parallel hash table takes $O(n)$ work
    and $O(\log^* n)$ span \whp.

\emph{Comparison sort:}
Cole presents a parallel merge sort that sorts $n$ elements in $O(n \log n)$ work and $O(\log n)$ span~\cite{cole1988parallel}.

\emph{Integer sort:}
    Suppose that we have $n$ non-negative integers in the range $[0, \on{poly}(n)]$.\footnote{$\on{poly}(n)$ means $O(n^c)$ for some constant $c$.}
    For any positive integer $q$, we can sort these integers in $O(qn)$ work and
    $O(qn^{1/q})$
    span~\cite{vishkin2010thinking}. Raman provides another integer
    sorting algorithm that runs in $O(n\log \log n)$ work and $O(\log n / \log
    \log n)$ span
    \whp~\cite{raman1990power}.

    Also, we can sort $n$ non-negative rational numbers
    whose numerators and denominators are bounded by $r \in \on{poly}(n)$ with the
    same asymptotic running times. Consider two distinct rational numbers $a/b$
    and $c/d$ that meet this criterion. Their absolute difference is
    $\abs{ad-bc}/\abs{bd} \ge 1/\abs{bd} \ge 1/r^2$.
    Therefore, if we multiply each rational number by $r^2$ and round them down to the
    nearest integer, we get $n$ integers bounded by $r^3 \in \on{poly}(n)$, whose
    sorted order matches the sorted order of the original rational numbers.

\emph{Graph connectivity:} Gazit gives an algorithm for graph
connectivity with $O(m+n)$ expected work and $O(\log n)$ span
\whp~\cite{gazit1991optimal}.

\emph{Matrix multiplication:}
Two $n$-by-$n$ matrices can be multiplied in $O(n^{\omega_p})$ work and $O(\log
n)$ span with parallel matrix multiplication constant $\omega_p \le 2.373$~\cite{dhulipala2021parallel}.

\section{Review of SCAN algorithms}

In this section, we provide an overview of the
SCAN~\cite{xu2007scan} and GS*-Index~\cite{wen2017efficient}
clustering algorithms.

\subsection{SCAN definitions}

The typical problem formulation for graph clustering is to
output a partition (or \emph{clustering}) of the vertices of the input graph
such that each cluster in the partition has many edges within the cluster and
there are few edges between clusters. How exactly to quantify the
quality of a clustering depends on the application domain.
\Cref{sec:clustering-quality-measures} lists two clustering quality measures.

SCAN~\cite{xu2007scan} is a graph clustering algorithm on undirected graphs.
The output of SCAN diverges slightly from this description of clustering in that SCAN
may leave some vertices unclustered. Unclustered vertices are further separated
into \emph{hubs} and \emph{outliers}. Hubs are unclustered vertices that neighbor multiple
clusters, and outliers are unclustered vertices that
neighbor at most one cluster.

\begin{figure}
  \includegraphics[scale=0.45]{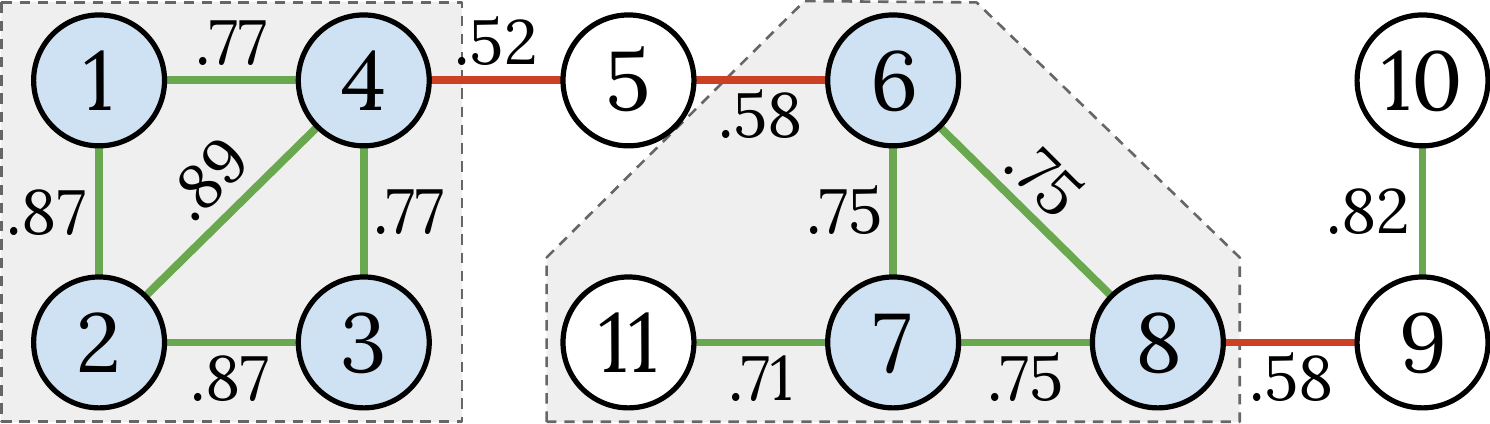}
  \caption{Example SCAN clustering with $\mu = 3$ and $\ep=.6$.
    The labels on the edges are cosine similarities. Core vertices are blue,
    whereas non-core vertices are white. Edges with similarity greater than $\ep$ are green,
    whereas other edges are red. There are two clusters (vertices $\set{1, 2, 3,
  4}$ and vertices $\set{6,7,8,11}$) as well as three unclustered vertices (hub
  vertex $5$ and outlier vertices $9$ and $10$).}
  \Description{Small example graph and the result of running SCAN on the graph.}
  \label{fig:scan-ex}
\end{figure}

For each pair of adjacent vertices $\set{u,v} \in E$, SCAN computes a similarity
score $\sigma(u,v)$. The original paper~\cite{xu2007scan} assumes that edges are unweighted and defines
the similarity score to be the cosine similarity of the closed neighborhoods of
the two vertices:
\[
  \sigma(u, v) = \on{CosineSim}(
    \clone(u),
    \clone(v)) =
  \frac{\smallabs{\clone(u) \cap
  \clone(v)}}{\sqrt{\smallabs{\clone(u)}}\sqrt{\smallabs{\clone(v)}}}
.\]
For instance, in the graph in \cref{fig:scan-ex}, the cosine similarity between
vertices $5$ and $6$ is
\begin{align*}
  \frac{\abs{\set{4,5,6} \cap
  \set{5,6,7,8}}}{\sqrt{\abs{\set{4,5,6}}}\sqrt{\abs{\set{5,6,7,8}}}}
  = \frac{2}{\sqrt{12}} \approx .58
.\end{align*}
This is just one possible choice for the similarity score, however. Other
papers consider using Jaccard similarity, Dice similarity, or weighted cosine
similarity for the similarity function~\cite{huang2010shrink,
huang2012revealing, chang2017pscan, mai2018scalable}.

SCAN takes two parameters as input, an integer $\mu \ge 2$ and a similarity
threshold $\ep \in [0, 1]$. Call vertices $u$ and $v$ \emph{$\ep$-similar} if
their similarity $\sigma(u, v)$ is at least $\ep$. The
\emph{$\ep$-neighborhood} $N_\eps(v)$ of a vertex $v$ is the set of its $\ep$-similar
neighbors,
  $\smallset{u \in \clone(v) \mid \sigma(u, v) \ge \ep}$.
The \emph{core} vertices are the vertices whose $\ep$-neighborhood contains at
least $\mu$ neighbors, i.e., vertices $v$ such that
  $\abs{N_\eps(v)} \ge \mu$.
A vertex $u$ is \emph{structurally reachable} from core vertex $v$ if there is a path of
vertices $v_1, v_2, \ldots, v_k$ for some $k \ge 2$ where $v_1 = v$, where $v_k
= u$, and where $v_i$ is a core and is $\ep$-similar to $v_{i+1}$ for each
integer $i$ from $1$ to $k-1$.

The two following properties define each cluster in the clustering that SCAN
finds:
\begin{itemize}[topsep=1pt,itemsep=0pt,parsep=0pt,leftmargin=15pt]
  \item The cluster is ``connected'': for any two vertices $u$ and
    $x$ in the cluster $C$, there is a vertex $v$ such that both $u$ and $x$ are
    structurally reachable from $v$.
  \item The cluster is ``maximal'': for every core vertex $v$ in
    the cluster, all vertices structurally reachable from $v$ are also
    in the cluster.
\end{itemize}
\Cref{fig:scan-ex} shows the clusters that result from running SCAN
on a small graph.

The $\emph{border}$
vertices, which are the clustered non-core vertices (e.g., vertex 11 in
\cref{fig:scan-ex}), may belong to several
distinct clusters according to the definition of SCAN clusters. The original SCAN algorithm
assigns each of these ambiguous border vertices to any of its possible
clusters arbitrarily.

Computing similarity scores takes $O(\alpha m)$ time with an appropriate
triangle counting algorithm;
to calculate a similarity score $\sigma(u, v)$, it suffices to count
the number of shared neighbors in $N(u) \cap N(v)$, which is precisely the
number of triangles in which edge $\set{u, v}$ appears.
After computing similarities, SCAN finds clusters by performing a modified
breadth-first search, which takes $O(n + m)$ time.

\subsection{Index-based SCAN: GS*-Index}\label{sec:prelims-gs}

GS*-Index~\cite{wen2017efficient} improves on SCAN by precomputing an index
from which finding cores and $\ep$-similar neighbors is fast for any setting
of $\mu$ and $\ep$. It takes $O((\alpha + \log n) m)$ time to compute the index,
and the index takes $O(m)$ space.  After computing the index, the time it takes
to compute the clustering for arbitrary query parameters $(\mu, \ep)$ depends on the size
of the resulting clusters rather than on the size of the full graph.
Specifically, for a subset of vertices $U \subseteq V$, define $E_{U, \ep}$ to be
the set of $\ep$-similar edges in the subgraph induced by $U$. Then the time to
compute the clustering $\mathcal{C}$ for parameters $\mu$ and $\ep$ is
$O\left(\abs{\bigcup_{U \in \mathcal{C}} E_{U, \ep}}\right)$. Determining whether
unclustered vertices are hubs or outliers is not considered in this time bound.

\begin{figure}
  \includegraphics[scale=0.38]{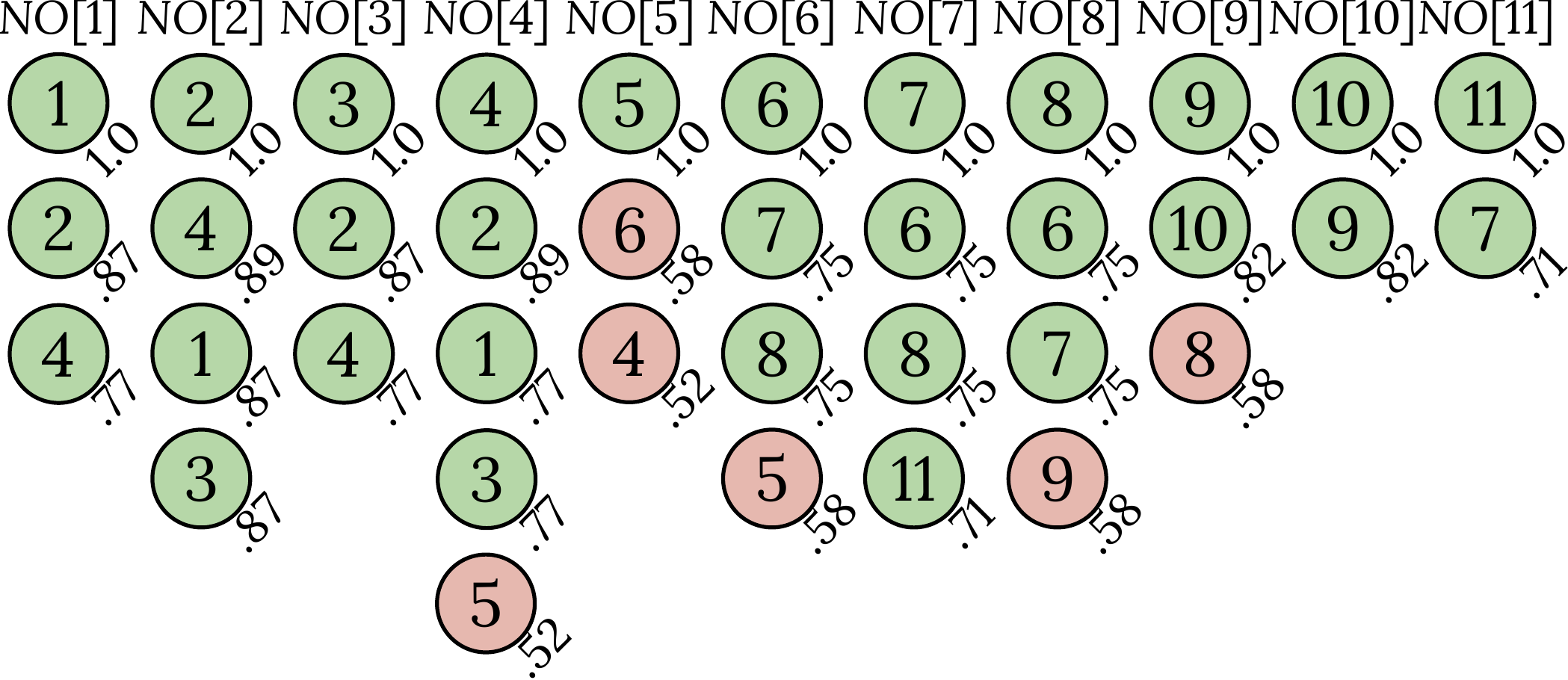}
  \caption{Neighbor order for the graph from \cref{fig:scan-ex}. In this figure, for each $v \in
    V$, we display $\nord[v]$
    as a column. The numbers beside each vertex are similarity scores. For example, in
    $\nord[3]$, the $.87$ label beside vertex $2$ represents the cosine similarity of
    $.87$ between vertices $3$ and $2$.
    Like in \cref{fig:scan-ex}, we consider the specific case where
  $\ep=0.6$ and color all $\ep$-similar neighbors green and all
    other neighbors red.}
  \Description{Neighbor order of a small example graph.}
  \label{fig:norder-ex}
\end{figure}
\begin{figure}
  \includegraphics[scale=0.36]{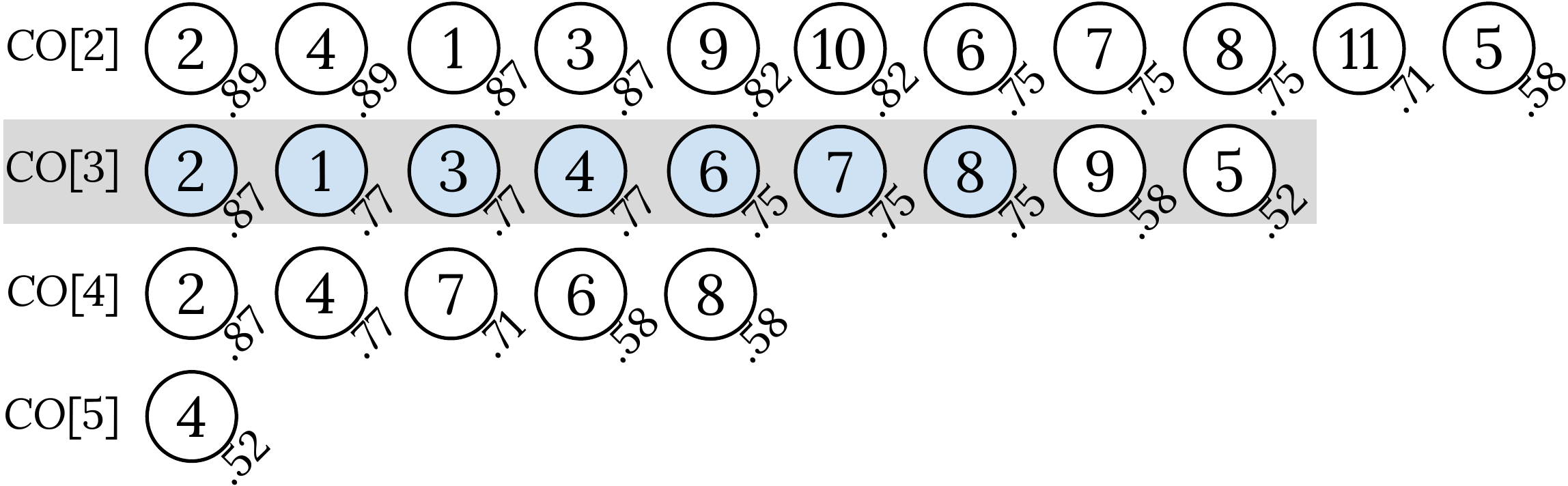}
  \caption{Core order for the graph from \cref{fig:scan-ex}. Each
    entry of $\cord$ is displayed horizontally. We omit $\cord[1]$ since
    we assume $\mu \ge 2$.
    The number beside each vertex in a row
    $\cord[\mu]$ is the core threshold for that vertex for that $\mu$. For
    example, in
    $\cord[2]$, the number $.75$ beside vertex 6 means that when $\mu = 2$ and $\ep \le .75$,
    vertex $6$ is a core vertex.
    Like in \cref{fig:scan-ex}, we consider the
    specific case where $(\mu, \ep) = (3, .6)$; we highlight $\cord[\mu]$ in gray and
    color the core vertices (i.e., vertices with core threshold at least
    $\ep$) blue.
  }
  \Description{Core order of a small example graph.}
  \label{fig:corder-ex}
\end{figure}

The index consists of two data structures, the \emph{neighbor order}
$\nord$ and the \emph{core order} $\cord$. To compute the index,
we first compute the similarity scores between every pair of adjacent vertices.
The neighbor order is the adjacency list of the graph with each neighbor
list sorted by non-increasing similarity. \Cref{fig:norder-ex} shows the neighbor
order for the graph in \cref{fig:scan-ex}. The core order is an array where the
$\mu$-th entry, $\cord[\mu]$, for any $\mu$ is a list of vertices with 
$\smallabs{\clone(\cdot)}$
at least $\mu$, i.e., all possible core vertices for this $\mu$ value. The vertices in $\cord[\mu]$ are sorted by
non-increasing similarity with vertex $\nord[\cdot][\mu]$. This similarity of a
vertex $v$ with vertex $\nord[v][\mu]$
is $v$'s \emph{core threshold} value. For any
$\ep$ no greater than the threshold, the
vertex is a core vertex under parameters $\mu$ and $\ep$.
\Cref{fig:corder-ex} shows the core order for the graph in \cref{fig:scan-ex}.
For example, to compute $\cord[3]$ in \cref{fig:corder-ex}, we consider the nine
vertices $\set{1, 2, 3, \ldots, 9}$ with $\smallabs{\clone(\cdot)} \ge 3$, determine their core
thresholds by looking at the similarities in the third row of
\cref{fig:norder-ex}, and sort the vertices by non-increasing core threshold.

To find the clustering resulting from SCAN parameters $\mu$ and $\ep$,
we perform a breadth-first search on the core vertices, considering only $\ep$-similar
edges in the graph and not searching further from any non-core vertices. The core
vertices and $\ep$-similar edges are easy to find from the index since the core
vertices are a prefix of $\cord[\mu]$ and the $\ep$-similar edges are prefixes
of each list in $\nord$ due to the sorting. The breadth-first search reveals all the SCAN clusters in
the graph.

\section{Parallel algorithm}\label{sec:alg-basic-description}

This section presents our new work-efficient, logarithmic-span parallel algorithms for
constructing the same SCAN index that GS*-Index constructs, and for retrieving
clusters from the index.

For the algorithm descriptions in this section, we assume the existence of basic utility
functions and of functions implementing the primitives listed in
\cref{sec:par-primitives}. The $\textsc{AllocateArray}(s)$ function
allocates an array that holds $s$ elements. The $\textsc{MakeHashMap}(\cdot)$ function
makes a hash table with the input argument specifying the key-value elements in the table.
The $\textsc{MakeHashSet}(\cdot)$ function makes a hash table
containing only keys rather than key-value pairs.
The $\textsc{Sum}(\cdot)$ function
returns the sum of the elements in an array via the reduce operation. The
$\textsc{RemoveDuplicates}(\cdot)$ function returns an array that has the same
set of elements that the input array has, but without any duplicate values.

\subsection{Index construction}\label{sec:index-alg}

\subsubsection{Computing similarities}\label{sec:compute-sims}

To shorten exposition, this section will only focus on one similarity function
$\sigma(\cdot, \cdot)$: cosine similarity for weighted graphs. Given a weighted
undirected graph $G = (V, E, w)$, the similarity score between two adjacent
vertices $\set{u, v}$ in $E$ is
\begin{align*}
  \sigma(u, v) &=
  \on{WeightedCosineSim}(\clone(u), \clone(v)) \\ &=
  \frac{\sum_{x \in \clone(u) \cap \clone(v)}
    w(u, x) w(v, x)
    }{\sqrt{\sum_{x \in \clone(u)} w(u, x)^2}\sqrt{\sum_{x \in \clone(v)} w(v, x)^2}}
\end{align*}
where we set $w(x, x) = 1$ for each vertex $x$. This weighted cosine similarity
measure is the natural generalization to the cosine similarity measure for
unweighted graphs that the original SCAN and GS*-Index algorithms consider.
Modifying the algorithm described in this section to instead compute the
unweighted cosine similarity or Jaccard similarity is straightforward.

\begin{algorithm}
  \caption{Algorithm for computing the cosine similarity of each edge in a
  weighted graph.}
  \label{alg:compute-sims}
\begin{algorithmic}[1]
  \figfont
  \Ensure An  array of length $m$ containing the similarity score of each edge.
  \Procedure{ComputeSimilarities}{$G=(V, E, w)$}
    \State $\cvar{norms} \gets \set{\sqrt{\sum_{u \in \clone(v)} w(u, v)} : v \in V}$
    \InlineComment
    Compute each entry of the array $\cvar{norms}$ with
    $\Call{Sum}{\cdot}$.
    \State $\cvar{similarities} \gets \Call{AllocateArray}{m}$
    \InlineComment For clarity, we index into $\cvar{similarities}$ with
      edges from $E$.
    \State $\cvar{neighbor\_tables} \gets \Call{AllocateArray}{n}$
        \label{algline:neighbor-tables}
    \Pfor{$v \in V$}       \label{algline:make-hash-loop}
      \State $\cvar{neighbor\_tables}[v] \gets \Call{MakeHashSet}{\clone(v)}$
      \label{algline:make-hash}
    \EndFor
    \Pfor{$\set{u, v} \in E$}
      \State (Without loss of generality, let $\smallabs{\clone(u)} \le \smallabs{\clone(v)}$.)
      \State $\cvar{shared\_neighbor\_weights} \gets \Call{AllocateArray}{\smallabs{\clone(u)}}$
      \For{$i \in \smallset{1, 2, 3, \ldots, \smallabs{\clone(u)}}$} \textbf{in parallel}
        \label{alg:compute-sims:lookup-begin}
        \State $x \gets $ $i$-th element in $\clone(u)$
        \State \begin{alglinebreaks}
            $\cvar{shared\_neighbor\_weights}[i] \gets$ \par
             \algind $w(u, x) \cdot w(v, x)$ \textbf{if} $x \in \cvar{neighbor\_tables}[v]$ \textbf{else} $0$
          \end{alglinebreaks}
        \label{alg:compute-sims:lookup-end}
      \EndFor
      \State \begin{alglinebreaks}
        $\cvar{similarities}[\set{u, v}] \gets$ \par
           \algind $\Call{Sum}{\cvar{shared\_neighbor\_weights}}
         / (\cvar{norms}[u] \cdot \cvar{norms}[v])$
      \end{alglinebreaks}
      \label{alg:compute-sims:compute-sim}
    \EndFor
    \State \Return $\cvar{similarities}$
  \EndProcedure
\end{algorithmic}
\end{algorithm}

\Cref{alg:compute-sims} gives pseudocode for computing similarities. The
logic follows that of a known parallel algorithm for triangle
counting~\cite{shun2015multicore}. The algorithm creates a hash set for each
vertex's neighborhood (\crefrange{algline:make-hash-loop}{algline:make-hash}). Then, for each pair of adjacent
vertices $u$ and $v$, looking up the neighbors of $u$ in the hash set for $v$'s
neighborhood (\crefrange{alg:compute-sims:lookup-begin}{alg:compute-sims:lookup-end})
gives all of the shared neighbors between $u$ and $v$, which allows
the algorithm to compute $\on{WeightedCosineSim}(\clone(u),
\clone(v))$ (\cref{alg:compute-sims:compute-sim}).

If the algorithm always searches for neighbors of the lower-degree vertex in the
hash set of the higher-degree vertex's neighborhood, the work is
$O\left(\sum_{\set{u,v}\in E} \min\smallset{\smallabs{\clone(u)},\smallabs{\clone(v)}}\right)$ in
expectation, which
is bounded by $O(\alpha m)$~\cite{chiba1985arboricity}. The span is $O(\log n)$
\whp.

For dense graphs, we can use matrix multiplication to obtain a work
bound of $O(n^{\omega_p})$. Let $W$ be an $n$-by-$n$ matrix with
$W_{u, v} = w(u,v)$ for arbitrary vertices $u$ and $v$. Then
$(W^2)_{u,v}$ is the numerator of $\on{WeightedCosineSim}(\clone(u),
\clone(v))$, so we can skip
\crefrange{algline:neighbor-tables}{alg:compute-sims:lookup-end} and
substitute $(W^2)_{u,v}$ for
$\Call{Sum}{\cvar{shared\_neighbor\_weights}}$ on
\cref{alg:compute-sims:compute-sim}.

\subsubsection{Neighbor order and core order}

\begin{algorithm}
  \caption{Algorithms for computing the neighbor order and core order.}
  \label{alg:compute-orders}
\begin{algorithmic}[1]
  \figfont
  \Procedure{MakeNeighborOrder}{$G=(V, E, w), \cvar{similarities}$}
    \State $\nord \gets \Call{AllocateArray}{n}$
    \Pfor{$v \in V$}
      \State $\nord[v] \gets \clone(v)$
        \label{alg:compute-orders:sort-adj-begin}
      \State
        Sort $u$ in $\nord[v]$ by
        non-increasing $\cvar{similarities}[\set{u, v}]$ value.
        \label{alg:compute-orders:sort-adj-end}
    \EndFor
    \State \Return $\nord$
    \EndProcedure
    \medskip
  \Procedure{MakeCoreOrder}{$G=(V, E, w), \nord$}
  \State $\cvar{sorted\_V} \gets$ $V$ sorted by non-increasing degree.
      \label{alg:compute-cores:sort-v}
    \State $\cvar{max\_degree} \gets \max_{v \in V}\smallabs{\clone(v)}$
    \State $\cord \gets \Call{AllocateArray}{\cvar{max\_degree}}$
    \Pfor{$\mu = \set{2, 3, 4, \ldots, \cvar{max\_degree}}$}
      \State $\cord[\mu] \gets \smallset{v \in V \mid \smallabs{\clone(v)} \ge \mu}$
      \InlineComment
        Find $\smallset{v \in V \mid \smallabs{\clone(v)} \ge \mu}$ by doubling
        search on $\cvar{sorted\_V}$.
        \label{alg:compute-cores:find-cores}
      \State Sort $v$ in $\cord[\mu]$ by non-increasing $\cvar{similarities}[\set{v, \nord[v][\mu]}]$ value.
        \label{alg:compute-cores:sort-cores}
    \EndFor
    \State \Return $\cord[\mu]$
  \EndProcedure
\end{algorithmic}
\end{algorithm}

After computing all similarity values, we construct the neighbor order and core
order (\cref{alg:compute-orders}). We form the neighbor order by
sorting each vertex's neighbor list by non-increasing
similarity
(\crefrange{alg:compute-orders:sort-adj-begin}{alg:compute-orders:sort-adj-end}).
Then, we form the core order by, for each $\mu$ value, finding all
possible core vertices under parameter $\mu$ (\cref{alg:compute-cores:find-cores})
and sorting them by non-increasing core threshold
(\cref{alg:compute-cores:sort-cores}).
On \cref{alg:compute-cores:find-cores}, to find all possible core vertices
(i.e., all vertices $v$ such that $\smallabs{\clone(v)} \ge \mu$), we perform a
doubling search on $\cvar{sorted\_V}$, the set of vertices sorted by
non-increasing degree (\cref{alg:compute-cores:sort-v}). This doubling search consists
of sequentially searching for the minimum $i$, such that the $2^i$-th entry of $\cvar{sorted\_V}$
fails to satisfy the predicate $\smallabs{\clone(\cdot)} \ge \mu$, and then performing
binary search on the last interval of the doubling search. Doubling search
is needed for optimal work bounds. Using only binary search would add
$O(n \log n)$ in total to the work since each binary search costs $O(\log n)$
work. Doubling search, on the other hand, costs only $O(\log j)$ work to find an
item located at index $j$. The $O(\log j)$ cost is also better than the $O(j)$
work and span that linear search would incur.

With a work-efficient comparison sort algorithm, the work analysis is
the same as the original analysis for GS*-Index, giving bounds of
$O(m \log n)$ work and $O(\log n)$ span for constructing
the orders.

If the graph is unweighted, each Jaccard similarity is a
rational number, and each unweighted cosine similarity squared is a rational
number. Recall from \cref{sec:par-primitives} that we can sort
rational numbers with an integer sorting algorithm. Therefore, if the graph is
unweighted, we can achieve better work bounds by
using integer sorting rather than comparison sorting.
 In order to apply the integer sort running
time bounds directly when computing the neighbor order, instead of sorting
$\nord[v]$ separately for each $v \in V$ like \cref{alg:compute-orders}
describes, we instead prepend $v$ to every entry in $\nord[v]$ for each $v \in
V$ and sort all elements in $\nord$ with a single integer sort. We perform the same
transformation to compute the core order with one integer sort. By doing this,
the complexity for computing the neighbor order and core
order match the complexity for integer sort on $m$ integers, as described in
\cref{sec:par-primitives}.

Summing the bounds for computing similarities with the bounds for
constructing the neighbor and core order gives the following theorems.
\begin{theorem}
  Fix an undirected, weighted graph and let $\alpha$ be its arboricity.
  Running the parallel SCAN index construction algorithm on the graph using cosine similarity
  as the similarity measure runs in $O((\alpha + \log n)m)$
  work (matching the work bound of GS*-Index) and $O(\log n)$ span \whp.
\end{theorem}
\begin{theorem}
  Fix an undirected, unweighted graph and let $\alpha$ be the graph's arboricity.
  The parallel SCAN index construction algorithm with cosine similarity
  or Jaccard similarity as the similarity measure can achieve the following
  running time bounds depending on what integer sorting algorithm is used:
  \begin{itemize}[topsep=1pt,itemsep=0pt,parsep=0pt,leftmargin=15pt]
    \item $O((\alpha + \log \log n)m)$ work and $O(\log n)$ span \whp,
    \item $O(\alpha m)$ work and $O(n^\beta)$ span \whp for any
       $0<\beta\leq 1$.
  \end{itemize}
\end{theorem}
In both theorems, we can replace the $\alpha m$ work term with $n^{\omega_p}$
if we use matrix multiplication to compute similarities.

\subsection{Querying for clusters}\label{sec:query-alg}

Next, we describe an efficient parallel algorithm for discovering clusters given
the parameters $\mu$ and $\epsilon$. The algorithm uses the index structure from
\cref{sec:index-alg}.

\begin{algorithm}
  \caption{Helper function for finding core vertices under a particular setting
  of SCAN parameters.}
  \label{alg:get-cores}
  \begin{algorithmic}[1]
  \figfont
    \Ensure An array of core vertices under SCAN parameters $(\mu, \ep)$.
    \Procedure{GetCores}{$\mu, \ep, \nord, \cord, \cvar{similarities}$}
      \State $\cvar{max\_degree} \gets \abs{\cord}$
      \If{$\mu > \cvar{max\_degree}$}
        \ \Return $\set{}$ \Comment No vertices are cores.
      \Else
        \ \Return $\set{v \in \cord[\mu] \mid
        \cvar{similarities}[\set{v, \nord[v][\mu]}] \ge \ep}$
          \label{algline:bin-search-cores}
          \InlineComment  \begin{alglinebreaks}
            Find cores
            using a doubling search on $\cord[\mu]$.
        \end{alglinebreaks}
      \EndIf
    \EndProcedure
  \end{algorithmic}
\end{algorithm}
\begin{algorithm}
  \caption{Helper function for assigning border non-core vertices to clusters after
  clustering all of the core vertices.}
  \label{alg:assign-non-cores}
  \begin{algorithmic}[1]
  \figfont
    \Procedure{AssignNonCores}{$\cvar{similar\_edges}$, $\cvar{cores\_set}$,
      $\cvar{clusters}$}
      \State $\cvar{subgraph\_vertices} \gets
        \Call{RemoveDuplicates}{\set{v \mid \set{u, v} \in \cvar{similar\_edges}}}$\label{algline:assign-non-cores:get-non-cores-begin}

      \State
        $\cvar{subgraph\_non\_cores} \gets
         \set{v \in \cvar{subgraph\_vertices} \mid v \not\in
          \cvar{cores\_set}}$
       \Comment Filter
        \label{algline:assign-non-cores:get-non-cores-end}
      \State $\cvar{non\_cores\_count} \gets \abs{\cvar{subgraph\_non\_cores}}$
      \State  $\cvar{assignments} \gets
        \Call{AllocateArray}{\cvar{non\_cores\_count}}$
       \State   $\cvar{non\_core\_indices} \gets
         \Call{MakeHashMap}{\set{\cvar{subgraph\_non\_cores}[i] \mapsto i}}$
      \Pfor{$i \in \set{1, 2, 3, \ldots, \cvar{non\_cores\_count}}$}
        \State $\cvar{assignments}[i] = \cvar{null}$
        \label{algline:assign-non-cores:preprocess-end}
      \EndFor
      \Pfor{$\set{u, v} \in \cvar{similar\_edges} \land
        (u \not\in \cvar{cores\_set} \lor
          v \not\in \cvar{cores\_set})$}
          \State  (Without loss of generality, let $v
            \not\in \cvar{cores\_set}$;  then, $u \in \cvar{cores\_set}$.)
        \State   $\cvar{address} \gets
          \&(\cvar{assignments}[
              \cvar{non\_core\_indices}[v]])$
        \State $\Call{CompareAndSwap}{\cvar{address}, \cvar{null},
        \cvar{clusters}[u]} $
          \label{algline:assign-non-cores:cas}
          \InlineComment  \begin{alglinebreaks}
            Assign border vertex $v$ to an arbitrary neighboring $\epsilon$-similar cluster. \par
            If the CAS fails, then that means $v$ is already assigned.
        \end{alglinebreaks}
      \EndFor
        \State
        \begin{alglinebreaks}
        For $v$ in $\cvar{subgraph\_non\_cores}$ in parallel, insert
        \par\algind $[v \mapsto \cvar{assignments}[\cvar{non\_core\_indices}[v]]]$
          into $\cvar{clusters}$.
        \end{alglinebreaks}
        \State \Return $\cvar{clusters}$
    \EndProcedure
  \end{algorithmic}
\end{algorithm}
\begin{algorithm}
  \caption{Algorithm for finding the SCAN clustering with parameters $\mu$ and
  $\ep$ from the index.}
  \label{alg:cluster}
  \begin{algorithmic}[1]
  \figfont
    \Procedure{Cluster}{$\mu, \ep, \nord, \cord, \cvar{similarities}$}
      \State $\cvar{cores} \gets \Call{GetCores}{\mu, \ep, \nord, \cord, \cvar{similarities}}$
        \label{algline:cluster:fetch-subgraph-begin}
      \State $\cvar{cores\_set} \gets \Call{MakeHashSet}{\cvar{cores}}$
      \State  $\cvar{similar\_edges} \gets
        \set{ \set{u, v} \mid u \in \cvar{cores\_set} \land \cvar{similarities}[\set{u,
          v}] \ge \ep}$
          \label{algline:bin-search-edges}
          \InlineComment  Get $\cvar{similar\_edges}$ by doubling search
            on $\nord[u]$ for each $u \in \cvar{cores}$.
        \label{algline:cluster:fetch-subgraph-end}
      \State
      \begin{alglinebreaks}
      $\cvar{similar\_core\_edges} \gets$ \par
      \algind $\set{ \set{u, v} \in
        \cvar{similar\_edges} \mid u \in \cvar{cores\_set} \land
        v \in \cvar{cores\_set}}$
        \quad\Comment Filter
      \end{alglinebreaks}
      \label{algline:cluster:filter-non-cores}
      \State
       \begin{alglinebreaks}
      $\cvar{core\_clusters} \gets$ Connected components of subgraph induced by \par
      \algind $\cvar{similar\_core\_edges}$, represented as a hash table mapping \par
      \algind $[v \mapsto \text{component ID}]$ for each $v \in \cvar{cores}$.
        \end{alglinebreaks}
      \label{algline:cluster:conn}
      \State \Return $\Call{AssignNonCores}{\cvar{similar\_edges},
        \cvar{cores\_set}, \cvar{core\_clusters}}$
    \EndProcedure
  \end{algorithmic}
\end{algorithm}

\Cref{alg:cluster} provides pseudocode for extracting a clustering with
arbitrary parameters from the index. \Cref{alg:get-cores,alg:assign-non-cores}
are subroutines for \cref{alg:cluster}.
To retrieve the clustering with parameters $\mu$ and $\ep$,
the algorithm performs a doubling search on $\cord[\mu]$ to find all core vertices
(\cref{algline:bin-search-cores} of \cref{alg:get-cores}) and then
performs doubling searches on $\nord[v]$ for each core vertex $v$
to find all $\ep$-similar edges incident on core vertices (\cref{algline:bin-search-edges} of \cref{alg:cluster}). For instance, for the graph in
\cref{fig:scan-ex} with parameters $(\mu, \ep) = (3, .6)$, the search on
$\cord[\mu]$ finds the blue vertices in \cref{fig:corder-ex}, and the
searches on $\nord[\cdot]$ find the green vertices in \cref{fig:norder-ex}.
This corresponds exactly to the blue core vertices and green edges in
\cref{fig:scan-ex}.

For each of these prefixes of
$\nord[v]$, the algorithm also creates a copy with all border non-core neighbors (e.g.,
vertex 11 in \cref{fig:scan-ex}) filtered away
(\cref{algline:cluster:filter-non-cores} of \cref{alg:cluster}). These filtered
prefixes constitute an adjacency list for the subgraph induced by the
$\ep$-similar edges on the core vertices. Running a parallel connectivity algorithm
on this subgraph assigns all core vertices to a cluster
(\cref{algline:cluster:conn} of \cref{alg:cluster}). Finally, the algorithm
takes all of the border non-core neighbors
(\crefrange{algline:assign-non-cores:get-non-cores-begin}
{algline:assign-non-cores:get-non-cores-end} of \cref{alg:assign-non-cores})
and uses compare-and-swap to assign each of them to
the same cluster as an
arbitrary neighboring $\ep$-similar core (\cref{algline:assign-non-cores:cas} of
\cref{alg:assign-non-cores}). The final output is a hash table mapping vertices
to cluster IDs. The algorithm achieves the bounds stated in the
following theorem.
\begin{theorem}
  Suppose the clustering algorithm, \cref{alg:cluster}, runs and returns a
  collection of clusters $\mathcal{C}$. For a set of vertices $U \in
  \mathcal{C}$, define $E_{U, \ep}$ to be the set of $\ep$-similar edges in the
  subgraph induced by $U$. Define $Z = \abs{\bigcup_{U \in \mathcal{C}} E_{U, \ep}} \in O(m)$.
  Then the clustering algorithm runs in $O(Z)$ expected work (which
  matches the work bound for GS*-Index) and $O(\log
  n)$ span \whp.
\end{theorem}
This theorem holds because the doubling searches in \crefrange{algline:cluster:fetch-subgraph-begin}
{algline:cluster:fetch-subgraph-end} of \cref{alg:cluster} fetch all of the edges
$\bigcup_{U \in \mathcal{C}} E_{U, \ep}$ in the output clustering $\mathcal{C}$
in a work-efficient manner,
and the remainder of the clustering algorithm operates only on the subgraph
given by $\bigcup_{U \in \mathcal{C}} E_{U, \ep}$ in a work-efficient
manner.

\subsection{Determining hubs and outliers}

After finding a clustering, we can determine whether unclustered vertices
are hubs or outliers. For each unclustered vertex $v$, we map each neighbor
in $N(v)$ to its cluster ID and reduce over the neighbors to
determine whether the vertex has neighbors belonging to distinct clusters. It
takes $O(\abs{N(v)})$ work and $O(\log \abs{N(v)})$ span to determine whether
$v$ is a hub or outlier. In total, this takes $O(m+n)$ work and $O(\log n)$ span.

\section{Approximating similarities}\label{sec:alg-approx}

After constructing the index, querying for a clustering is fast. Index
construction itself, though, may be expensive since it takes $\Omega(\on{min}\{\alpha m, n^{\omega}\})$
work. One unexplored technique for
speeding up SCAN is to use LSH
to approximate similarities.

For example, to use SimHash to approximate cosine similarities, we fix a sample
size $k \in \N$. Then, we draw $kn$ random numbers from the standard normal distribution, which is possible via
the Box-Muller transform~\cite{box1958note} given a source of uniform random
numbers. With these normally distributed random numbers, we construct a $k$-sample sketch of $\clone(v)$ for each
vertex $v$. The sketching takes $O(km)$ work and $O(\log n)$ span using
the reduce operation to compute inner products.
Now we can compute the similarity between any adjacent vertices $u$ and
$v$ by comparing their sketches in $O(k)$ work and $O(\log k)$ span.
Computing the sketches and the similarities over all edges takes $O(km)$ work
and $O(\log n + \log k)$ span. The work bound is better than
the work bound for computing exact similarities if $k$ is asymptotically less than the
arboricity $\alpha$. Similarly, we
can use MinHash to approximate Jaccard similarities.

We can then compute a neighbor order and core order based on these
similarities. Again, we can achieve better work bounds using an
integer sort algorithm, and in fact we can use integer sort on both
unweighted and weighted graphs. This is because the approximate
similarities are non-negative integers scaled by a factor of $\pi/k$
for SimHash or $1/k$ for MinHash, and we can postpone scaling the integers until
after sorting. Therefore, we can construct a SCAN index with the following
running time bounds.

\begin{theorem}
  Fix an undirected graph and let $k \le \on{poly}(n)$. The parallel SCAN index
  construction algorithm using $k$-sample MinHash (for unweighted graphs) or
  SimHash (for unweighted or weighted graphs) to compute approximate
  similarities can achieve the following running time bounds depending on what
  integer sorting algorithm is used:
  \begin{itemize}[topsep=1pt,itemsep=0pt,parsep=0pt,leftmargin=15pt]
    \item $O((k + \log \log n)m)$ work and $O(\log n)$ span \whp,
    \item $O(km)$ work and $O(n^{\beta})$ span for  $0<\beta\leq 1$.
  \end{itemize}
\end{theorem}

We can also theoretically analyze the clusterings that
result from using these approximate similarities. In particular, suppose we fix
some $\ep \in [0, 1]$ and $\delta \in (0, 1)$. Notice that the SCAN clustering with
parameters $\ep$ and arbitrary $\mu$ is only concerned about whether
similarities fall above or below $\ep$, rather than exact similarity values. If the number of samples is
sufficiently high, then \whp, all
edges outside the similarity range $\ep \pm \delta$ will be
``correctly classified'' as above or below the threshold $\ep$ by the approximate similarities.
We present such a result for approximating cosine similarity using SimHash.
\begin{theorem}
  \label{thm:simhash-accuracy-bound}
  Let $G = (V, E, w)$ be an undirected graph with non-negative edge weights, let $\ep \in [0,1]$, and let $\delta \in (0,1)$. Suppose
  $k \ge \pi^2 \ln(nm)/(2\delta^2)$
  and suppose we use SimHash with $k$ samples to
  compute approximate cosine similarity scores for every edge in $G$.
  Then \whp, all edges with exact cosine similarities outside
  the interval $(\ep - \delta, \ep + \sqrt{1 - \ep^2}\delta)$ are correctly
  classified by the approximate similarities as above or below the
  threshold $\ep$.
\end{theorem}
\begin{proof}
  Consider an arbitrary edge $\set{u, v} \in E$ with an exact cosine
  similarity $s \in [0, 1]$ outside the interval $(\ep - \delta, \ep + \sqrt{1 -
  \ep^2}\delta)$. It suffices to prove that the edge is correctly classified by
  the approximate cosine similarity with probability at least $1 - 1/(nm)$. Then,
  applying a union bound over all $m$ edges gives that all edges outside the similarity
  interval are classified correctly \whp.

  Let $\theta = \cos(s)$ be the angle between the vectors corresponding to $\clone(u)$ and
  $\clone(v)$. The angle is in the
  range $[0, \pi/2]$ since all edge weights are non-negative. Recall from
  \cref{sec:lsh} that the SimHash estimate for the
  angle between the two vectors is $\hat{\theta} \sim \on{Binomial}(k, \theta/\pi) \cdot \pi/k$.
  Hoeffding's inequality~\cite{hoeffding1963probability} implies that
  given arbitrary $\ell \in \N$, $p \in [0,1]$, and $t > 0$, for a binomial
  random variable $X \sim \on{Binomial}(\ell, p)$, the probabilities
  $\Pr[X/\ell \ge p + t]$ and $\Pr[X/\ell \le p - t]$ are each
  bounded above by $\exp(-2\ell t^2)$.
  Using this inequality on $\hat{\theta}$
  with $\ell = k$, $p = \theta/\pi$, and $t = \delta/\pi$
  gives that
  both $\Pr[\hat{\theta} \ge \theta + \delta]$ and $\Pr[\hat{\theta} \le \theta - \delta]$ are
  each bounded above by
    $\exp\left(
      -2k \delta^2/\pi^2
    \right)
    \le 1/(nm)$.

  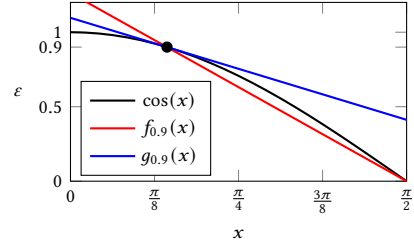
\begin{figure}\centering
  \begin{tikzpicture}
  \begin{axis}[
          width=0.34*\textwidth,
          height=0.18*\textheight,
          xmin=0, xmax={pi/2}, 
          ymin=0, ymax=1.2, 
          domain=0:{pi/2},
          legend pos=south west,
          xlabel={$x$},
          xtick={0, 0.39269908169, 0.78539816339, 1.1780972451, 1.57079632679},
          xticklabels={$0$, $\frac{\pi}{8}$, $\frac{\pi}{4}$, $\frac{3\pi}{8}$,
            $\frac{\pi}{2}$},
          ytick={0, 0.5, 0.9, 1}, 
          ylabel={$\ep$},
          ylabel style={rotate=-90},
          legend style={font=\figfont},
          label style={font=\figfont},
          tick label style={font=\figfont},
  ]
  \addplot[mark=none, thick, black]  {cos(deg(x))};
  \addplot[mark=none, thick, red]  {0.9 - 0.9*(x - 0.451027)/(1.57079632679 - 0.451027)};
  \addplot[mark=none, thick, blue]  {0.9 - sin(deg(0.451027))*(x - 0.451027)};
  \node[circle,fill,inner sep=1.5pt] at (axis cs:0.451027,0.9) {};
  \legend{$\cos(x)$,$f_{0.9}(x)$,$g_{0.9}(x)$}
  \end{axis}
  \end{tikzpicture}
  \caption{Plot of the SimHash approximation lower and upper bound functions on
  cosine for $\ep = 0.9$. The bold point is $(\phi, \eps)$.}
  \Description{
    Plot of the lower and upper bound functions that we use to bound
    $\cos(\phi + \delta)$ and $\cos(\phi - \delta)$.
  }
  \label{fig:simhash-cos}
\end{figure}

  Let $\phi = \arccos(\ep) \in [0,\pi/2]$ be the similarity threshold $\ep$ transformed into an angle
  threshold. First, consider the case where $s \in [0, \ep - \delta]$, which also
  implies that $\ep \ge \delta$. The
  straight line from
  the point $(\phi, \ep)$ to the point $(\pi/2, 0)$ has the equation
  \[
    f_\ep(x) = \ep - \frac{\ep}{\pi/2 - \phi}(x - \phi)
  ,\]
  which \cref{fig:simhash-cos} shows in red for $\ep=.9$.
  By concavity of the cosine function in $[0, \pi/2]$, we have that $\cos(x) \ge f_\ep(x)$ when $x
  \in [\phi, \pi/2]$. We have that $\phi + \delta =
  \arccos(\ep) + \delta \le \arccos(\delta) + \delta \le \pi/2$;
  the first inequality comes from the arccosine function being decreasing
  combined with the constraint that $\ep \ge \delta$, and
  the second inequality comes from taking derivatives to maximize
  $\arccos(\delta) + \delta$ for $\delta \in (0, 1)$. Therefore, we can
  substitute $\phi + \delta$ for $x$ in
  the inequality $\cos(x) \ge f_\ep(x)$ to get that
  $\cos(\phi+\delta) \ge \ep -
  \frac{\ep}{(\pi/2 - \arccos(\ep))}\delta$.  Plotting the multiplicative factor
  $\frac{\ep}{(\pi/2 - \arccos(\ep))}$ with varying $\ep$ shows that the
  factor falls in the range $ [2/\pi, 1]$, giving a looser but clearer
  bound that $\cos(\phi+\delta) \ge \ep - \delta \ge s$. Taking the arccosine of
  the leftmost and rightmost sides of the inequality gives $
  \phi+\delta
  = \arccos(\cos(\phi+\delta))
  \le
  \arccos(s) = \theta$, where the first equality uses the fact that $\phi + \delta \in [0,
  \pi/2]$. Now the upper bound on $\Pr[\hat{\theta} \le \theta - \delta]$
  gives that
  the probability that $\hat{\theta} > \theta - \delta \ge \phi$ is at least $1 - 1/(nm)$. Taking
  the cosine of both sides gives that the
  cosine similarity estimate $\cos(\hat{\theta})$ falls below $\ep$ with
  probability at least $1-1/(nm)$ as desired.

  Next, consider  the case where $s \in [\ep + \sqrt{1 - \ep^2}\delta, 1]$.
  If $\ep = 1$, then $s = 1$ and SimHash will always return the correct estimate
  $\cos(\hat{\theta}) = \cos(0) = 1$ as desired. For $\ep < 1$,
  define $h(\delta) = \left(1-\delta^2\right)/\left(1+\delta^2\right)$ and note
  that
  \begin{align*}
    \ep + \sqrt{1-\ep^2}\delta \le 1 \iff
    \delta \le \frac{1-\ep}{\sqrt{1-\ep^2}} \iff \\
    \delta^2 \le \frac{(1-\ep)^2}{1-\ep^2} =
    \frac{1-\ep}{1+\ep}  \iff
    \delta^2 + \ep\delta^2 \le 1-\ep \iff
    \ep \le h(\delta)
  \end{align*}
  Next, linearize the cosine function at the input point $\phi$ to get the line
  \begin{align*}
    g_\ep(x) &= \ep - \sin(\phi)(x - \phi)
    = \ep - \sin(\arccos(\ep))(x - \phi) \\
             &= \ep - \sqrt{1- \ep^2} (x - \phi)
  ,\end{align*}
  which \cref{fig:simhash-cos} shows in blue for $\ep=.9$.
  By concavity of the cosine function, we have that
  $\cos(x) \le g_\ep(x)$ when $x \in [0, \pi/ 2]$.
  Note that we have that $\phi - \delta = \arccos(\ep) - \delta
  \ge \arccos(h(\delta)) - \delta \ge 0$;
  the first inequality comes from the arccosine function being decreasing
  combined with the constraint that $\ep \le h(\delta)$, and
  the second inequality comes from plotting
  $\arccos(h(\delta)) - \delta$ to see that it is
  non-negative for $\delta \in (0, 1)$ .
  Hence, we can substitute $\phi - \delta$ for $x$ in the inequality $\cos(x) \le
  g_\ep(x)$ to get that
    $\cos(\phi - \delta) \le \ep + \sqrt{1 - \ep^2}\delta \le s$.
  Taking the arccosine of the leftmost and rightmost sides of the inequality
  gives $\phi-\delta
  = \arccos(\cos(\phi-\delta))
  \ge \arccos(s) = \theta$, where the first equality uses the fact that $\phi - \delta \in [0,
  \pi/2]$.
  Now, the upper bound on $\Pr[\hat{\theta} \ge \theta + \delta]$ gives that
  the probability that $\hat{\theta} <  \theta + \delta \le \phi$ is at least $1 - 1/(nm)$. Taking
  the cosine of both sides gives that the
  cosine similarity estimate $\cos(\hat{\theta})$ is above $\ep$ with
  probability at least $1-1/(nm)$ as desired.
\end{proof}
We also present a similar result for approximating the Jaccard similarity using MinHash.
\begin{theorem}
  \label{thm:minhash-accuracy-bound}
  Let $G = (V, E)$ be an undirected graph, let $\ep \in [0,1]$, and let $\delta
  \in (0,1)$. Suppose $k \ge \ln(nm)/\left(2\delta^2\right)$
  and suppose we use standard MinHash with $k$ samples to
  compute approximate Jaccard similarity scores for every edge in $G$ .
  Then \whp, all edges with exact Jaccard similarities outside
  the interval $(\ep - \delta, \ep + \delta)$ are correctly
  classified by the approximate similarities as above or below the
  threshold $\ep$.
\end{theorem}
\begin{proof}
  The result follows from applying Hoeffding's inequality as in the proof of
  \cref{thm:simhash-accuracy-bound}. We omit full details for brevity.
  %
  %
\end{proof}

Though the bounds in these theorems require a large number of samples $k$ to
achieve high accuracy, our experiments in \cref{chapter:experiments} show that
lower values of $k$ still achieve good clusterings.
This approximation strategy helps for denser graphs with large
arboricity.

\section{Implementation}\label{sec:alg-impl}

We implement the algorithms described in
\cref{sec:alg-basic-description,sec:alg-approx} to determine how they perform
in practice. We write our code in C++ within the Graph Based Benchmark Suite
(GBBS) framework~\cite{dhulipala2018theoretically,dhulipala2020graph}. GBBS
provides libraries useful for implementing parallel graph
algorithms. Our implementations use the concurrent hash table
implementation~\cite{shun2014phase}, parallel sorting algorithms, and various
graph processing helper functions that GBBS provides.

Though the algorithms as described in
\cref{sec:alg-basic-description,sec:alg-approx} achieve good
theoretical bounds, our actual implementations make several changes for better
performance. This section explains the more significant changes.

\subsection{Computing similarities}\label{sec:compute-sims-impl}

We implement similarity computation for both cosine similarity and Jaccard
similarity.
Experiments by Shun and Tangwongsan~\cite{shun2015multicore} suggest that
the hash-based approach to triangle counting or computing similarities in
\cref{alg:compute-sims} incurs many cache misses and that a merge-based
approach is faster in comparison, even though it increases the asymptotic
work bound from $O(m\alpha)$ to $O(m^{3/2})$.
Our implementation uses the merge-based approach of Shun and Tangwongsan.
This approach assumes that each neighbor list in the adjacency list of the
input graph is sorted by vertex number, which is true for graphs converted to
GBBS's graph file format. In order to count each triangle only once and
hence reduce work, we construct a directed version of the input
graph by filtering each neighbor list so as to direct each edge towards its
higher-degree vertex.  Then, for each pair of adjacent vertices $(u, v)$,
we find triangles of the form $\set{(u, v), (v, x), (u, x)}$ for $x$ in $N(u) \cap
N(v)$ by merging the out-neighborhoods of $u$ and $v$ in the directed graph.
To get similarity scores for each pair of adjacent vertices, the implementation
maintains an atomic counter for each edge and increments the counters for
all three edges of any triangle found.

The merge logic between two neighbor lists follows the logic of the parallel
merge implementation in GBBS: if both neighbor lists are small, we iterate through
the sorted neighbor lists sequentially to find shared neighbors; if one neighbor
list is small and the other is large, then we search for each element of the small
neighbor list in the larger list via binary search; and finally, if both neighbor lists are
large, then we split them into smaller sub-lists and recursively merge the
sub-lists in parallel.

To compute similarities using matrix multiplication instead of merge, we use
the Intel Math Kernel Library's \code{cblas\_sgemm} function
for matrix multiplication.\footnote{\url{
https://software.intel.com/content/www/us/en/develop/tools/oneapi/components/onemkl.html}}
Though its documentation does not provide asymptotic running time bounds, it
runs well in practice.

\subsection{Querying for clusters}

When querying
the index for clusters (\cref{alg:cluster}), we find the connected components on
the core vertices (\cref{algline:cluster:conn}) by
using the concurrent union-find implementation from the GBBS
codebase~\cite{dhulipala2021connectit}.
Using union-find
allows us to avoid materializing the subgraph to pass to a work-efficient connectivity
algorithm.  We ``union'' the edges in
$\cvar{similar\_core\_edges}$ (\cref{algline:cluster:filter-non-cores}) and
apply ``find'' to each vertex to
populate an $n$-length array of vertices' cluster IDs rather than a hash table
as described in \cref{algline:cluster:conn}.
Having this array also
simplifies the logic for $\textsc{AssignNonCores}$ (\cref{alg:assign-non-cores})
by changing $\textsc{AssignNonCores}$ to skip the preprocessing in
lines~\ref{algline:assign-non-cores:get-non-cores-begin}--\ref{algline:assign-non-cores:preprocess-end}
and instead compare-and-swap directly into the cluster ID array.

\subsection{Approximate similarities}\label{sec:approx-sims-impl}

We implement similarity approximation logic using both SimHash and MinHash. For
MinHash, we use a variant called $k$-partition MinHash, or one permutation
hashing~\cite{li2012one}.
It is more computationally efficient than the original version of
MinHash; computing a sketch of a vertex $v$ takes only $O(k +
\smallabs{\clone(v)})$ work using $k$-partition MinHash, rather than
$O(k\smallabs{\clone(v)})$ work using standard MinHash, since $k$-partition
MinHash generates a $k$-length sketch using only one permutation rather than $k$
permutations. The $k$-partition
variant still provides reasonable clustering results, but the accuracy bound
in \cref{thm:minhash-accuracy-bound} no longer applies for this variant.

When the number of samples $k$ for the LSH approximation scheme
is high, it becomes more expensive to compute and process sketches
to get approximate similarities. For low-degree vertices, the
merge-based algorithm described in \cref{sec:compute-sims-impl} is cheap and cache-friendly enough that
it is better to compute similarities exactly rather than approximately. As
a simple example, if two adjacent vertices have degree
significantly less than $k$, it is faster and more accurate to process the original
neighbor lists of the vertices instead of their $k$-length sketches.
To avoid sketching low-degree vertices, we add a heuristic to choose which
vertices to sketch and which similarities to approximate. The heuristic is to
only approximate similarities between pairs of vertices that both have sufficiently high
degree and to compute similarities exactly with triangle counting for all other
pairs of vertices. We determine whether a vertex is high degree by checking
whether its degree exceeds a threshold value of $k$ for
approximate cosine similarity and $3k/2$ for approximate Jaccard similarity. No
sketches are needed for vertices that either do not have high degree or do not
have any neighbors with high degree.

\section{Experiments}\label{chapter:experiments}

Our timing experiments show that our implementation achieves good
speedup over sequential baselines and performs
competitively against ppSCAN~\cite{che2018parallelizing}, a state-of-the-art parallel
shared-memory SCAN implementation. Our results for our approximate
SCAN implementation suggest that LSH can speed
up index construction while maintaining good clustering quality.

Non-shared-memory parallel algorithms fail to outperform our implementation as well.
Zhao et al.'s reported timings for their distributed SCAN
algorithm~\cite{zhao2013pscan} are much slower than our times;
they report taking 36 minutes with fifteen eight-core machines to
cluster their largest graph, which has four million vertices and
sixty million edges, whereas our algorithm takes under three seconds to process
the larger, denser Orkut graph using fewer cores.
Chen et al.~\cite{chen2013pscan} and Zhou and Wang~\cite{zhou2015sparkscan} only test
their distributed SCAN algorithms on graphs with fewer than two million edges and do not report
times. Stovall et al.~\cite{stovall2014gpuscan} only test their GPU-based SCAN algorithm on graphs with
fewer than six million edges, whereas we focus on much larger graphs in our
experiments.

\subsection{Benchmarking environment}

We run experiments on an Amazon EC2 \texttt{c5.24xlarge}
instance, which has 192 GiB of RAM and 48 CPU cores with two-way hyper-threading
for a maximum of 96 hyper-threads. We enable hyper-threading in our parallel experiments by default.
We implement our parallel algorithm using the merge-based approach for computing
similarities described in \cref{sec:compute-sims-impl} ({\ourscan} in the experimental plots)
as well as using matrix multiplication ({\ourscan}-MM in the plots).
We compare our parallel algorithm
using all 48 cores with hyper-threading to our algorithm using only 1 thread, to the
original sequential GS*-Index implementation,\footnote{We obtained the GS*-Index code via personal
correspondence with its authors.} and to
ppSCAN with AVX2 instructions\footnote{The ppSCAN code is available at
\url{https://github.com/RapidsAtHKUST/ppSCAN/tree/master/pSCAN-refactor}.} using all 48 cores with hyper-threading.
ppSCAN's authors show that ppSCAN outperforms other existing parallel SCAN
algorithms (pSCAN~\cite{chang2017pscan},
anySCAN~\cite{mai2018scalable}, and SCAN-XP~\cite{takahashi2017scan}).
For fixed parameters $\mu$ and $\ep$, all of these algorithms return
the same output, except that ambiguous border vertices may have
different assignments. All code written is C++ code, and compiled with GCC 7.5.0 using
the \code{-O3} optimization flag.
The {\ourscan}  code uses GBBS's scheduler library~\cite{blelloch2020parlaylib}
written using standard C++ threads. We run the parallel codes with \code{numactl
  -{}-interleave=all}, which interleaves memory allocations across
CPUs and gives better performance for this particular problem on the
EC2 instance.
Each time measurement is the median of five trials, unless specified otherwise.

\begin{table}
  {\figfont
  \begin{tabular}{llll}
  \toprule
  Name & Number of vertices & Number of edges & Type \\
  \midrule
  Orkut & 3,072,441 & 117,185,083 & unweighted \\
  brain & 784,262 & 267,844,669 & unweighted \\
  WebBase & 118,142,155 & 854,809,761 & unweighted \\
  Friendster & 65,608,366 & 1,806,067,135 & unweighted \\
  blood vessel & 25,825 & 70,240,269 & weighted \\
  cochlea & 25,825 & 282,977,319 & weighted \\
  \bottomrule
  \end{tabular}
  }
    \caption[Graph used in experiments]{Summary of the graphs for the experiments.}
  \label{tab:experiment-graphs}
\end{table}

\Cref{tab:experiment-graphs} summarizes the graphs that we use in the experiments.
``Orkut'' and ``Friendster'' are the com-Orkut and com-Friendster graphs
from the Stanford Large Network Dataset
Collection~\cite{snapnets}.\footnote{\url{https://snap.stanford.edu/data/}}
Both are social network graphs in which the vertices are users and the edges
represent friend relationships. ``Brain'' is the bn-human-Jung2015-M87113878
dataset from NeuroData\footnote{\url{https://neurodata.io/}} provided by
Network
Repository~\cite{networkrepository}.\footnote{\url{http://networkrepository.com/bn-human-Jung2015-M87113878.php}}
The graph represents a mapping of human brain connections. ``WebBase'' is the
webbase-2001 graph from the Laboratory for Web
Algorithmics~\cite{lawgraphs1,lawgraphs2}.\footnote{\url{http://law.di.unimi.it/webdata/webbase-2001/}}
The graph represents the links discovered by a web crawler. Although the original WebBase graph
is a directed graph, we change the edges to be
undirected and remove self-loop edges
so that SCAN can operate on the graph. ``Blood vessel'' and
``cochlea'' are weighted graphs from
HumanBase~\cite{greene2015understanding}.\footnote{\url{https://hb.flatironinstitute.org/download}
under the ``top edges'' column}  Vertices represent genes, edges represent pairs
of genes with evidence of a functional relationship in blood vessel tissues or
cochlea tissues, and edge weights represent the probability of there being a
 relationship. For computational convenience, on the brain,
Friendster, blood vessel, and cochlea graphs, we compact vertex IDs so that all IDs
are contiguous with no zero-degree vertices.

Neither GS*-Index and ppSCAN run on weighted graphs, so we do not run
them on the blood vessel and cochlea graphs. We also only test cosine
similarity on the weighted graphs since we did not implement weighted Jaccard
similarity for {\ourscan}.

\subsection{Clustering quality measures}\label{sec:clustering-quality-measures}
We evaluate our clustering results using the modularity and adjusted
Rand index measures. These quality measures are popular and consistent with
existing graph clustering literature.
%
The \emph{modularity} of a clustering is the proportion of edges that
fall within clusters in the clustering minus the expected number of
edges that would fall within clusters in a random graph with the same
degree distribution~\cite{newman2004finding}. Specifically, fix some
clustering, let $A_{u,v}$ for arbitrary vertices $u$ and $v$ be $1$ if $u$ and
$v$ are neighbors and be $0$ otherwise, and let $\delta_{u,v}$
be $1$ if $u$ and $v$ are assigned the same cluster and be $0$
otherwise. The modularity is computed as
\begin{align*}
  \frac{1}{2m}\sum_{u,v \in V}\left( A_{u,v} - \frac{\abs{N(u)}\abs{N(v)}}{2m} \right)\delta_{u,v}
.\end{align*}
The definition of modularity also easily extends to
weighted graphs~\cite{newman2004analysis}.
Higher modularity scores suggest better clusterings.

Another way to measure the quality of a proposed clustering is to check how similar it is
against a known ground-truth clustering. The \emph{adjusted Rand index}
(ARI)~\cite{hubert1985comparing} is one well-known metric for evaluating
this similarity.
ARI counts the number of pairs of vertices, such that the two vertices are assigned to the same
clusters or to different clusters in both the proposed clustering and the
ground-truth clustering. This count is then adjusted for chance.
Let $\mathcal{C}$ be the proposed clustering on the set of $n$
vertices $V$ and let $\mathcal{G}$ be the ground-truth clustering. For
integers $i$ in $\set{1, 2, 3, \ldots, \abs{\mathcal{C}}}$ and $j$ in
$\set{1, 2, 3, \ldots, \abs{\mathcal{G}}}$, let $n_{i, j}$ be the
number of vertices in both cluster $i$ of $\mathcal{C}$ and cluster
$j$ of $\mathcal{G}$. Let $n_{i,*} =
\sum_{j=1}^{\abs{\mathcal{G}}} n_{i, j}$ and let $n_{*,j} =
\sum_{i=1}^{\abs{\mathcal{C}}} n_{i, j}$ for each $i$ and $j$.
Then, the ARI between $\mathcal{C}$ and
$\mathcal{G}$ is
\begin{align*}
  \frac{
    \sum_{i=1}^{\abs{\mathcal{C}}}
    \sum_{j=1}^{\abs{\mathcal{G}}}
    \binom{n_{i,j}}{2}
    -
    \sum_{i=1}^{\abs{\mathcal{C}}} \binom{n_{i,*}}{2}
    \sum_{j=1}^{\abs{\mathcal{G}}} \binom{n_{*,j}}{2}
    / \binom{n}{2}
  }{
    \left(
      \sum_{i=1}^{\abs{\mathcal{C}}} \binom{n_{i,*}}{2}
      +
      \sum_{j=1}^{\abs{\mathcal{G}}} \binom{n_{*,j}}{2}
    \right) / 2
    -
    \sum_{i=1}^{\abs{\mathcal{C}}} \binom{n_{i,*}}{2}
    \sum_{j=1}^{\abs{\mathcal{G}}} \binom{n_{*,j}}{2}
    / \binom{n}{2}
  }
.\end{align*}
Higher ARI scores suggest a better match with the ground-truth clustering.
Neither the modularity nor the ARI can exceed $1$, and they may
be negative if the clustering is ``worse than random.''

\subsection{Results}

\subsubsection{Index construction time comparison}

The first experiment measures the running time to construct the SCAN index with
exact cosine similarity. The time to compute the index using Jaccard
similarity is about the same (at most $9$\% difference for {\ourscan} on 48 cores), so we do not report it
separately.
\Cref{fig:index-construction-times} shows the time measurements.
{\ourscan} achieves a parallel
self-relative speedup of 23--70$\times$ on index construction. Moreover,
{\ourscan} running sequentially is 1.4-2.2$\times$ faster than the original
GS*-Index implementation, likely due to the directed triangle counting
optimization that \cref{sec:compute-sims-impl} describes, so the speedup of
{\ourscan} on 48 cores with hyper-threading is 50--151$\times$ over GS*-Index.
On the two dense graphs with fewer vertices, {\ourscan}-MM outperforms
{\ourscan}, but it takes too much memory to run on the other graphs.

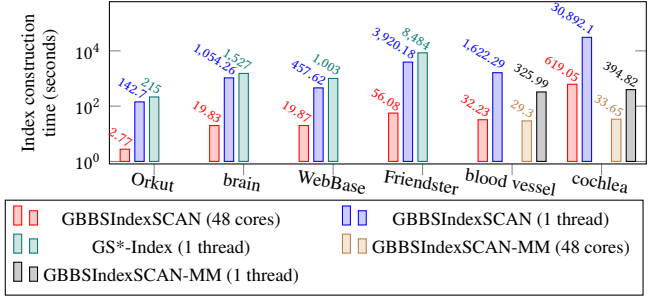
\begin{figure} \centering
\begin{tikzpicture}
\begin{axis}[
    width=0.49*\textwidth,
    bar width=0.007*\textwidth,
    height=0.17*\textheight,
    ybar,
    enlarge x limits=0.1, 
    legend style={
      at={(0.40,-.22)},
      anchor=north,
      font=\figfont,
      /tikz/every even column/.append style={column sep=0.5cm} 
    },
    legend columns=2,
    ylabel style={align=center}, 
    ylabel={Index construction\\ time (seconds)},
    ymode=log,
    ymax=600000, 
    point meta=rawy, 
    every node near coord/.append style={rotate=-32, anchor=east, xshift=4pt,
      yshift=3pt, font=\tiny}, 
    symbolic x coords={Orkut,brain,WebBase,Friendster,blood vessel,cochlea},
    xtick=data,
    xticklabel style={rotate=-10, anchor=north, yshift=3pt}, 
    xtick pos=bottom, 
    nodes near coords,
    nodes near coords align={vertical},
    label style={font=\figfont},
    tick label style={font=\figfont},
    ]
    \addplot[color=red,fill=red!20!white] coordinates {(Orkut, 2.76594) (brain, 19.8275) (WebBase, 19.8692)
        (Friendster, 56.0829) (blood vessel, 32.2283) (cochlea, 619.052)};
    \addplot[color=blue,fill=blue!20!white] coordinates {(Orkut, 142.697) (brain, 1054.26) (WebBase, 457.622)
        (Friendster, 3920.18) (blood vessel, 1622.29) (cochlea, 30892.1)};
    \addplot[color=PineGreen,fill=PineGreen!20!white] coordinates {(Orkut, 215) (brain, 1527) (WebBase, 1003) (Friendster, 8484)};
    \addplot[color=brown,fill=brown!20!white] coordinates {(blood vessel,
      29.2982) (cochlea, 33.6489)};
    \addplot[color=black,fill=black!20!white] coordinates {(blood vessel,
      325.9944) (cochlea, 394.8165)};
    \legend{{\ourscan} (48 cores),{\ourscan} (1 thread), GS*-Index
    (1 thread), {\ourscan-MM} (48 cores), {\ourscan-MM} (1 thread) }
\end{axis}
\end{tikzpicture}
\caption[Exact index construction times]{Index construction times with exact cosine
  similarity as the similarity measure. We only run {\ourscan}-MM on the blood
  vessel and cochlea graphs, whose
  adjacency matrices fit in memory.}
\Description{
  Plot displaying the construction times for {\ourscan} and GS*-Index.
  We only run {\ourscan}-MM on the small dense graphs, whose
  adjacency matrices fit in memory.
}
\label{fig:index-construction-times}
\end{figure}


\subsubsection{Clustering time comparison}

The second experiment is to measure the running time for querying for the
clustering over various settings of parameters $(\mu, \eps)$. The plots
only consider exact cosine similarity since times are about the same
using Jaccard similarity (at most either $10^{-4}$ absolute difference or $4$\%
difference for {\ourscan} on 48 cores). Clustering behavior is the same
between {\ourscan} and {\ourscan}-MM, so we omit times for {\ourscan}-MM.
\Cref{fig:vary-eps-query-times} measures the running times with $\mu = 5$
and $\ep \in \set{.1,.2,.3,\ldots, .9}$, and \cref{fig:vary-mu-query-times} measures the running
times with $\ep = 0.6$ and $\mu \in \set{2, 4, 8, 16, \ldots, \min\set{16384,
2^{\floor{\log_2(\text{max degree})}}}}$.

\begin{figure*} \centering
  \begin{tikzpicture}
  \begin{groupplot}[
    group style =
      {group size=6 by 1,
       x descriptions at=edge bottom,
       ylabels at=edge left,
       horizontal sep=1cm,
     },
    height=0.15*\textheight,
    xlabel={$\ep$},
    xticklabel style={rotate=-45, anchor=west},
    x label style={at={(axis description cs:0.5,-0.2)}}, 
    ylabel={Query time (seconds)},
    ymode=log,
    width=0.18*\textwidth,
    legend columns=-1, 
    legend style={font=\figfont, /tikz/every even column/.append style={column sep=0.5cm}}, 
    label style={font=\figfont},
    tick label style={font=\figfont},
    title style={font=\figfont, yshift=-6pt},
  ]
  \nextgroupplot[title={Orkut},
    legend to name=legend-query-times]

      \addplot[mark options={solid,scale=0.5},mark=*,color=red] coordinates {
          (.1, .094003)
          (.2, .0586429)
          (.3, .028589)
          (.4, .0112741)
          (.5, .00304484)
          (.6, .000813007)
          (.7, .000353813)
          (.8, .000229836)
          (.9, .000138998)
      };
      \addlegendentry[color=black]{{\ourscan} (48 cores)}
      \addplot[mark options={solid,scale=1},mark=x,color=blue, densely dashed] coordinates {
        (.1, 4.65577)
        (.2, 2.42888)
        (.3, 1.11539)
        (.4, 0.398127)
        (.5, 0.097518)
        (.6, 0.0151381)
        (.7, 0.00278211)
        (.8, 0.000995159)
        (.9, 0.000819921)
      };
      \addlegendentry[color=black]{{\ourscan} (1 thread)}
      \addplot[mark options={solid,scale=1},mark=star,color=PineGreen,  densely dotted] coordinates {
          (.1, 1.45216)
          (.2, 0.917822)
          (.3, 0.48386)
          (.4, 0.178069)
          (.5, 0.05292)
          (.6, 0.012546)
          (.7, 0.002131)
          (.8, 0.001236)
          (.9, 0.00114)
      };
    \addlegendentry[color=black]{GS*-Index (1 thread)}
      \addplot[mark options={solid,scale=0.4},mark=square*,color=violet,dashdotted] coordinates {
          (.1, 1.267)
          (.2, 1.707)
          (.3, 1.566)
          (.4, 1.474)
          (.5, 1.129)
          (.6, 0.885)
          (.7, 0.62)
          (.8, 0.383)
          (.9, 0.27)
      };
    \addlegendentry[color=black]{ppSCAN (48 cores)}

  \nextgroupplot[title={brain}]
      \coordinate (c2) at (rel axis cs:1,1);

      \addplot[mark options={solid,scale=0.5},mark=*,color=red] coordinates {
          (.1, 0.234539)
          (.2, 0.203397)
          (.3, 0.142956)
          (.4, 0.096628)
          (.5, 0.0599639)
          (.6, 0.0339301)
          (.7, 0.0183439)
          (.8, 0.00739193)
          (.9, 0.000435829)
      };
      \addplot[mark options={solid,scale=1},mark=x,color=blue, densely dashed] coordinates {
          (.1, 10.2138)
          (.2, 8.54672)
          (.3, 6.09467)
          (.4, 3.9603)
          (.5, 2.401)
          (.6, 1.37229)
          (.7, 0.698837)
          (.8, 0.248032)
          (.9, 0.00476193)
        };
      \addplot[mark options={solid,scale=1},mark=star,color=PineGreen,  densely dotted] coordinates {
          (.1, 2.66912)
          (.2, 2.50678)
          (.3, 2.00977)
          (.4, 1.49183)
          (.5, 0.988608)
          (.6, 0.582534)
          (.7, 0.311483)
          (.8, 0.126319)
          (.9, 0.003197)
      };
      \addplot[mark options={solid,scale=0.4},mark=square*,color=violet,dashdotted] coordinates {
          (.1, 0.305)
          (.2, 0.418)
          (.3, 0.547)
          (.4, 0.64)
          (.5, 0.786)
          (.6, 0.893)
          (.7, 0.994)
          (.8, 0.83)
          (.9, 0.381)
      };

  \nextgroupplot[title={WebBase}]
    \addplot[mark options={solid,scale=0.5},mark=*,color=red] coordinates {
        (.1, 1.81441)
        (.2, 1.23065)
        (.3, 0.9109)
        (.4, 0.74103)
        (.5, 0.562619)
        (.6, 0.419559)
        (.7, 0.336332)
        (.8, 0.238046)
        (.9, 0.171985)
    };
    \addplot[mark options={solid,scale=1},mark=x,color=blue, densely dashed] coordinates {
        (.1, 72.0221)
        (.2, 50.4092)
        (.3, 37.1366)
        (.4, 31.4352)
        (.5, 22.5599)
        (.6, 16.489)
        (.7, 13.4312)
        (.8, 9.28345)
        (.9, 6.3021)
    };
    \addplot[mark options={solid,scale=1},mark=star,color=PineGreen,  densely dotted] coordinates {
        (.1, 16.0125)
        (.2, 12.6492)
        (.3, 10.6763)
        (.4, 8.93198)
        (.5, 7.33843)
        (.6, 5.8731)
        (.7, 4.66403)
        (.8, 3.50799)
        (.9, 2.5007)
    };
    \addplot[mark options={solid,scale=0.4},mark=square*,color=violet,dashdotted] coordinates {
        (.1, 2.291)
        (.2, 2.207)
        (.3, 2.037)
        (.4, 1.859)
        (.5, 1.751)
        (.6, 1.6)
        (.7, 1.496)
        (.8, 1.34)
        (.9, 1.215)
    };
  \nextgroupplot[title={Friendster}]
    \addplot[mark options={solid,scale=0.5},mark=*,color=red] coordinates {
        (.1, 1.42153)
        (.2, 0.407775)
        (.3, 0.092788)
        (.4, 0.0249019)
        (.5, 0.00943184)
        (.6, 0.00520706)
        (.7, 0.00429606)
        (.8, 0.00391889)
        (.9, 0.00378609)
    };
    \addplot[mark options={solid,scale=1},mark=x,color=blue, densely dashed] coordinates {
        (.1, 59.2489)
        (.2, 17.011)
        (.3, 3.83045)
        (.4, 0.911656)
        (.5, 0.286193)
        (.6, 0.103567)
        (.7, 0.0616119)
        (.8, 0.046922)
        (.9, 0.0408812)
    };
    \addplot[mark options={solid,scale=1},mark=star,color=PineGreen,  densely dotted] coordinates {
        (.1, 37.8245)
        (.2, 10.0391)
        (.3, 2.41182)
        (.4, 0.669428)
        (.5, 0.217547)
        (.6, 0.086502)
        (.7, 0.055531)
        (.8, 0.044358)
        (.9, 0.040409)
    };
    \addplot[mark options={solid,scale=0.4},mark=square*,color=violet,dashdotted] coordinates {
        (.1, 19.217)
        (.2, 17.19)
        (.3, 13.817)
        (.4, 11.066)
        (.5, 8.71)
        (.6, 6.845)
        (.7, 5.29)
        (.8, 3.408)
        (.9, 2.381)
    };
  \nextgroupplot[title={blood vessel}]
    \addplot[mark options={solid,scale=0.5},mark=*,color=red] coordinates {
        (.1, 0.0513239)
        (.2, 0.053705)
        (.3, 0.0398822)
        (.4, 0.022121)
        (.5, 0.0134902)
        (.6, 0.00805497)
        (.7, 0.00286102)
        (.8, 0.00137997)
        (.9, 0.000495195)
    };
    \addplot[mark options={solid,scale=1},mark=x,color=blue, densely dashed] coordinates {
        (.1, 1.44023)
        (.2, 1.43673)
        (.3, 1.14266)
        (.4, 0.567973)
        (.5, 0.283482)
        (.6, 0.131377)
        (.7, 0.0448461)
        (.8, 0.0137229)
        (.9, 0.00206804)
    };
  \nextgroupplot[title={cochlea}]
    \addplot[mark options={solid,scale=0.5},mark=*,color=red] coordinates {
        (.1, 0.192167)
        (.2, 0.192404)
        (.3, 0.191734)
        (.4, 0.1922)
        (.5, 0.189405)
        (.6, 0.180294)
        (.7, 0.153685)
        (.8, 0.075783)
        (.9, 0.0135601)
    };
    \addplot[mark options={solid,scale=1},mark=x,color=blue, densely dashed] coordinates {
        (.1, 8.22381)
        (.2, 8.20898)
        (.3, 8.22242)
        (.4, 8.21482)
        (.5, 8.10654)
        (.6, 7.58341)
        (.7, 6.31862)
        (.8, 2.88354)
        (.9, 0.44201)
    };
  \end{groupplot}
  \node at ($(group c3r1)!.5!(group c4r1) + (0, 1.45cm)$) {\ref*{legend-query-times}};
  \end{tikzpicture}
  \caption[Clustering query time with $\mu = 5$]{Clustering time with $\mu=5$
  and varying $\ep$ using exact cosine similarity as the similarity measure.}
  \Description{
    Plot displaying the time to cluster the graph with $\mu$ set to $5$ and with
    varying $\ep$ using {\ourscan} with 48 cores, {\ourscan} with one
    thread, GS*-Index with one thread, and ppSCAN with 48 cores.
  }
  \label{fig:vary-eps-query-times}
\end{figure*}
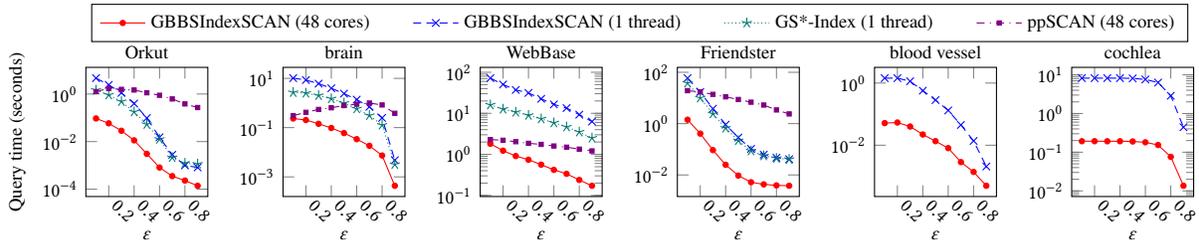

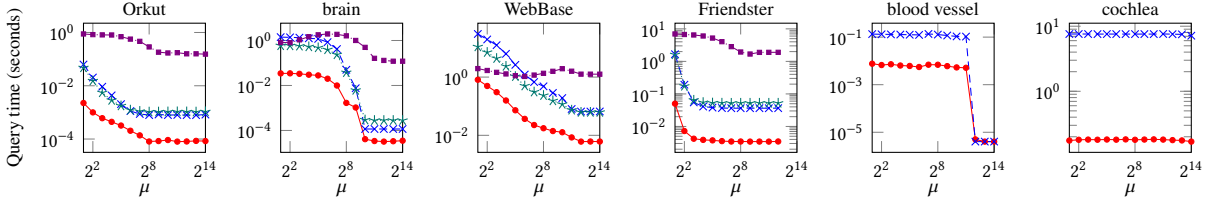
\begin{figure*} \centering
  \begin{tikzpicture}
  \begin{groupplot}[
    group style =
      {group size=6 by 1,
       x descriptions at=edge bottom,
       ylabels at=edge left,
       horizontal sep=1cm,
     },
    height=0.15*\textheight,
    xlabel={$\mu$},
    x label style={at={(axis description cs:0.5,-0.18)}}, 
    ylabel={Query time (seconds)},
    xmode=log,
    ymode=log,
    log basis x={2},
    xmin=2,
    xmax=16384,
    width=0.18*\textwidth,
    label style={font=\figfont},
    tick label style={font=\figfont},
    title style={font=\figfont, yshift=-6pt},
  ]
  \nextgroupplot[title={Orkut}]
      \addplot[mark options={solid,scale=0.5},mark=*,color=red] coordinates {
          (2, 0.0022552)
          (4, 0.00100017)
          (8, 0.000605822)
          (16, 0.000442028)
          (32, 0.000320911)
          (64, 0.000204802)
          (128, 0.000137091)
          (256, 8.08239e-05)
          (512, 8.39233e-05)
          (1024, 9.17912e-05)
          (2048, 8.01086e-05)
          (4096, 8.10623e-05)
          (8192, 8.51154e-05)
          (16384, 8.39233e-05)
      };
      \addplot[mark options={solid,scale=1},mark=x,color=blue, densely dashed] coordinates {
          (2, 0.0627329)
          (4, 0.0203891)
          (8, 0.0091989)
          (16, 0.00434899)
          (32, 0.00204206)
          (64, 0.00112486)
          (128, 0.000872135)
          (256, 0.00079298)
          (512, 0.000833988)
          (1024, 0.000794888)
          (2048, 0.000804901)
          (4096, 0.000791788)
          (8192, 0.000790119)
          (16384, 0.000826836)
      };
      \addplot[mark options={solid,scale=1},mark=star,color=PineGreen,  densely dotted] coordinates {
          (2, 0.047576)
          (4, 0.015254)
          (8, 0.005498)
          (16, 0.00284)
          (32, 0.00176)
          (64, 0.001223)
          (128, 0.001095)
          (256, 0.001033)
          (512, 0.001036)
          (1024, 0.001046)
          (2048, 0.001053)
          (4096, 0.001047)
          (8192, 0.001034)
          (16384, 0.001058)
      };
      \addplot[mark options={solid,scale=0.4},mark=square*,color=violet,dashdotted] coordinates {
          (2, 0.88)
          (4, 0.831)
          (8, 0.812)
          (16, 0.805)
          (32, 0.694)
          (64, 0.564)
          (128, 0.465)
          (256, 0.287)
          (512, 0.182)
          (1024, 0.17)
          (2048, 0.174)
          (4096, 0.158)
          (8192, 0.158)
          (16384, 0.154)
      };

  \nextgroupplot[title={brain}]
      \addplot[mark options={solid,scale=0.5},mark=*,color=red] coordinates {
          (2, 0.034415)
          (4, 0.0345809)
          (8, 0.0330091)
          (16, 0.030129)
          (32, 0.0279591)
          (64, 0.019978)
          (128, 0.00995016)
          (256, 0.00167894)
          (512, 0.00104618)
          (1024, 4.1008e-05)
          (2048, 3.38554e-05)
          (4096, 3.19481e-05)
          (8192, 3.31402e-05)
          (16384, 3.50475e-05)
      };
      \addplot[mark options={solid,scale=1},mark=x,color=blue, densely dashed] coordinates {
          (2, 1.41357)
          (4, 1.381)
          (8, 1.33615)
          (16, 1.25116)
          (32, 1.18129)
          (64, 0.84403)
          (128, 0.416078)
          (256, 0.049818)
          (512, 0.00754905)
          (1024, 0.00011611)
          (2048, 0.000115871)
          (4096, 0.00011611)
          (8192, 0.000115871)
          (16384, 0.00011611)
      };
      \addplot[mark options={solid,scale=1},mark=star,color=PineGreen,  densely dotted] coordinates {
          (2, 0.590159)
          (4, 0.582925)
          (8, 0.563906)
          (16, 0.526638)
          (32, 0.476063)
          (64, 0.38972)
          (128, 0.234019)
          (256, 0.03756)
          (512, 0.006122)
          (1024, 0.000293)
          (2048, 0.000278)
          (4096, 0.000278)
          (8192, 0.000279)
          (16384, 0.000279)
      };
      \addplot[mark options={solid,scale=0.4},mark=square*,color=violet,dashdotted] coordinates {
          (2, 0.818)
          (4, 0.803)
          (8, 1.069)
          (16, 1.478)
          (32, 1.786)
          (64, 1.953)
          (128, 1.875)
          (256, 1.61)
          (512, 1.003)
          (1024, 0.495)
          (2048, 0.161)
          (4096, 0.13)
          (8192, 0.12)
          (16384, 0.12)
      };

  \nextgroupplot[title={WebBase}]
    \addplot[mark options={solid,scale=0.5},mark=*,color=red] coordinates {
        (2, 0.810357)
        (4, 0.495299)
        (8, 0.303517)
        (16, 0.157027)
        (32, 0.072403)
        (64, 0.036696)
        (128, 0.0225549)
        (256, 0.0175118)
        (512, 0.0140271)
        (1024, 0.012743)
        (2048, 0.00846505)
        (4096, 0.00608706)
        (8192, 0.00612402)
        (16384, 0.00611711)
    };
    \addplot[mark options={solid,scale=1},mark=x,color=blue, densely dashed] coordinates {
        (2, 30.279)
        (4, 19.5896)
        (8, 12.4906)
        (16, 6.29713)
        (32, 2.63183)
        (64, 1.22212)
        (128, 0.677306)
        (256, 0.4812)
        (512, 0.286208)
        (1024, 0.190375)
        (2048, 0.0801871)
        (4096, 0.065706)
        (8192, 0.0656209)
        (16384, 0.0659809)
    };
    \addplot[mark options={solid,scale=1},mark=star,color=PineGreen,  densely dotted] coordinates {
        (2, 10.9597)
        (4, 7.03563)
        (8, 4.36268)
        (16, 2.34549)
        (32, 1.04628)
        (64, 0.496756)
        (128, 0.307526)
        (256, 0.233531)
        (512, 0.153503)
        (1024, 0.116425)
        (2048, 0.073369)
        (4096, 0.061627)
        (8192, 0.0617)
        (16384, 0.061691)
    };
    \addplot[mark options={solid,scale=0.4},mark=square*,color=violet,dashdotted] coordinates {
        (2, 1.957)
        (4, 1.685)
        (8, 1.455)
        (16, 1.282)
        (32, 1.117)
        (64, 1.051)
        (128, 1.172)
        (256, 1.348)
        (512, 1.58)
        (1024, 1.939)
        (2048, 1.592)
        (4096, 1.262)
        (8192, 1.26)
        (16384, 1.258)
    };
  \nextgroupplot[title={Friendster}]
    \addplot[mark options={solid,scale=0.5},mark=*,color=red] coordinates {
        (2, 0.0502429)
        (4, 0.00724387)
        (8, 0.00420213)
        (16, 0.00387192)
        (32, 0.00373602)
        (64, 0.0035491)
        (128, 0.0034492)
        (256, 0.00345111)
        (512, 0.00342798)
        (1024, 0.00343299)
        (2048, 0.00342703)
        (4096, 0.00345492)
    };
    \addplot[mark options={solid,scale=1},mark=x,color=blue, densely dashed] coordinates {
        (2, 1.6619)
        (4, 0.191906)
        (8, 0.0548241)
        (16, 0.0409329)
        (32, 0.0380552)
        (64, 0.0367911)
        (128, 0.036411)
        (256, 0.0364749)
        (512, 0.036478)
        (1024, 0.0365179)
        (2048, 0.0363429)
        (4096, 0.0365331)
    };
    \addplot[mark options={solid,scale=1},mark=star,color=PineGreen,  densely dotted] coordinates {
        (2, 1.6031)
        (4, 0.172846)
        (8, 0.064127)
        (16, 0.054367)
        (32, 0.052655)
        (64, 0.052061)
        (128, 0.05175)
        (256, 0.051832)
        (512, 0.051846)
        (1024, 0.051863)
        (2048, 0.051825)
        (4096, 0.0519)
    };
    \addplot[mark options={solid,scale=0.4},mark=square*,color=violet,dashdotted] coordinates {
        (2, 6.97)
        (4, 6.71)
        (8, 6.487)
        (16, 6.036)
        (32, 5.214)
        (64, 4.036)
        (128, 2.839)
        (256, 1.943)
        (512, 1.684)
        (1024, 1.892)
        (2048, 1.899)
        (4096, 1.901)
    };
  \nextgroupplot[title={blood vessel}]
    \addplot[mark options={solid,scale=0.5},mark=*,color=red] coordinates {
        (2, 0.00760794)
        (4, 0.00669789)
        (8, 0.00721598)
        (16, 0.00639582)
        (32, 0.00613189)
        (64, 0.00553989)
        (128, 0.00691795)
        (256, 0.00689292)
        (512, 0.00602603)
        (1024, 0.00532198)
        (2048, 0.00513697)
        (4096, 5.00679e-06)
        (8192, 4.05312e-06)
        (16384, 4.05312e-06)
    };
    \addplot[mark options={solid,scale=1},mark=x,color=blue, densely dashed] coordinates {
        (2, 0.134032)
        (4, 0.131393)
        (8, 0.130024)
        (16, 0.128199)
        (32, 0.126141)
        (64, 0.123702)
        (128, 0.135714)
        (256, 0.127727)
        (512, 0.116931)
        (1024, 0.106885)
        (2048, 0.105283)
        (4096, 4.05312e-06)
        (8192, 4.05312e-06)
        (16384, 4.05312e-06)
    };
  \nextgroupplot[title={cochlea}]
    \addplot[mark options={solid,scale=0.5},mark=*,color=red] coordinates {
        (2, 0.174993)
        (4, 0.179801)
        (8, 0.181049)
        (16, 0.180141)
        (32, 0.18148)
        (64, 0.180423)
        (128, 0.182211)
        (256, 0.179838)
        (512, 0.181261)
        (1024, 0.182456)
        (2048, 0.180912)
        (4096, 0.177656)
        (8192, 0.175623)
        (16384, 0.16859)
    };
    \addplot[mark options={solid,scale=1},mark=x,color=blue, densely dashed] coordinates {
        (2, 7.58723)
        (4, 7.59598)
        (8, 7.58198)
        (16, 7.56601)
        (32, 7.57966)
        (64, 7.57226)
        (128, 7.57069)
        (256, 7.56987)
        (512, 7.5735)
        (1024, 7.56765)
        (2048, 7.60462)
        (4096, 7.54639)
        (8192, 7.49629)
        (16384, 7.14768)
    };
  \end{groupplot}
  \end{tikzpicture}
  \caption[Clustering query time with $\ep = 0.6$]{Clustering time with $\ep=0.6$
  and varying $\mu$ using exact cosine similarity as the similarity measure.}
  \Description{
    Plot displaying the time to cluster the graph with varying $\mu$ and with
    $\mu$ set to 0.6 using {\ourscan} with 48 cores, {\ourscan} with one
    thread, GS*-Index with one thread, and ppSCAN with 48 cores.
  }
  \label{fig:vary-mu-query-times}
\end{figure*}

{\ourscan} is faster than ppSCAN and GS*-Index on all tested parameter
settings, though this ignores the time that
{\ourscan} takes to precompute its index.
This precomputation cost that {\ourscan} incurs is preferable over ppSCAN only when
the user makes many queries. Notably, though, it might not
take many queries for {\ourscan} to become preferable over ppSCAN. For
example, on the Orkut and Friendster graphs, the sum of the time
measurements for ppSCAN on the nine parameter settings in \cref{fig:vary-eps-query-times}
exceeds the sum of the corresponding measurements for {\ourscan} plus
the index construction time for {\ourscan}.

Sequentially, {\ourscan} can be slower at querying for clusters than
GS*-Index due to adjustments made
in {\ourscan} to make it more friendly to parallelism, namely using
union-find instead of sequential breadth-first search, as well as
iterating over all edges an additional time to assign non-core
vertices (\cref{alg:assign-non-cores}). It is up to 4.5$\times$ slower
than GS*-Index on the tested parameter settings and graphs. {\ourscan}
running on 48 cores, however, is faster than the other
implementations on all tested parameter settings; it is faster than
GS*-Index by 5--32$\times$ and faster than ppSCAN by
1.26--12,070$\times$.

\subsubsection{Approximate index construction time}

The third experiment measures the running time of constructing {\ourscan}
with 48 cores using the approximate cosine and approximate Jaccard similarity
measures with varying numbers of samples. For the weighted graphs, we only test
the approximate cosine similarity measure since the $k$-partition MinHash
variant that we implement for Jaccard similarity does not handle weighted graphs. Each trial uses a different pseudorandom seed for the
randomness in the approximation schemes.
\Cref{fig:approx-index-construction-times} displays the results.
The approximate Jaccard similarity implementation is consistently faster
than the approximate cosine similarity implementation because of the better
efficiency of constructing sketches for $k$-partition MinHash compared to
SimHash. The times plateau or even decrease at large sample sizes for some
graphs due to the heuristic discussed in \cref{sec:approx-sims-impl} for
avoiding sketching for low-degree vertices.

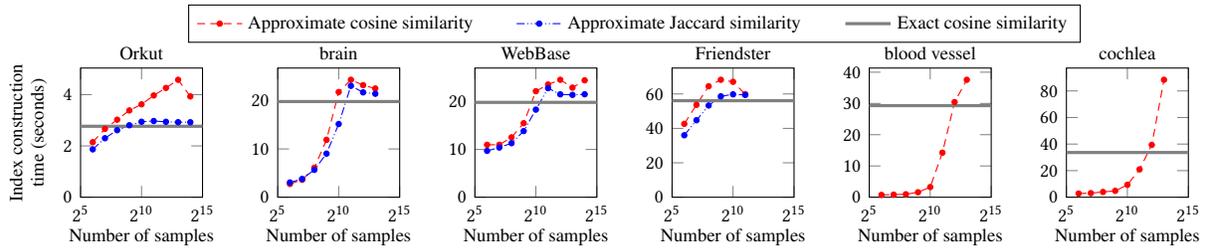
\begin{figure*} \centering
\begin{tikzpicture}
  \begin{groupplot}[
    group style =
      {group size=6 by 1,
       x descriptions at=edge bottom,
       ylabels at=edge left,
       horizontal sep=1cm,
     },
    height=0.15*\textheight,
    xlabel={Number of samples},
    x label style={at={(axis description cs:0.5,-0.18)}}, 
    ylabel style={align=center}, 
    ylabel={Index construction\\ time (seconds)},
    xmode=log,
    log basis x={2},
    xmin=32,
    xmax=32768,
    ymin=0,
    width=0.18*\textwidth,
    legend columns=-1, 
    legend style={font=\figfont,/tikz/every even column/.append style={column
      sep=0.5cm}}, 
    label style={font=\figfont},
    tick label style={font=\figfont},
    title style={font=\figfont, yshift=-6pt},
  ]
  \nextgroupplot[title={Orkut},
    legend to name=legend-approx-index-times]

      \addplot[mark options={solid,scale=0.5},mark=*,color=red, densely dashed] coordinates {
          (64, 2.14565)
          (128, 2.6652)
          (256, 3.02249)
          (512, 3.38223)
          (1024, 3.62116)
          (2048, 3.96687)
          (4096, 4.26101)
          (8192, 4.5783)
          (16384, 3.92902)
      };
      \addlegendentry[color=black]{Approximate cosine similarity}
      \addplot[mark options={solid,scale=0.5},mark=*,color=blue, densely dash dot dot] coordinates {
          (64, 1.86581)
          (128, 2.29971)
          (256, 2.61267)
          (512, 2.80486)
          (1024, 2.94262)
          (2048, 2.97076)
          (4096, 2.94041)
          (8192, 2.92462)
          (16384, 2.92243)
      };
      \addlegendentry[color=black]{Approximate Jaccard similarity}
      \addplot[color=gray, style={very thick}] coordinates {
          (32, 2.76594)
          (32768, 2.76594)
      };
      \addlegendentry[color=black]{Exact cosine similarity}

  \nextgroupplot[title={brain}]

      \addplot[mark options={solid,scale=0.5},mark=*,color=red, densely dashed] coordinates {
          (64, 2.74868)
          (128, 3.60887)
          (256, 6.10136)
          (512, 11.899)
          (1024, 21.8003)
          (2048, 24.3374)
          (4096, 23.2177)
          (8192, 22.5325)
      };
      \addplot[mark options={solid,scale=0.5},mark=*,color=blue, densely dash dot dot] coordinates {
          (64, 3.05867)
          (128, 3.79266)
          (256, 5.67183)
          (512, 9.03514)
          (1024, 15.1697)
          (2048, 23.0833)
          (4096, 21.729)
          (8192, 21.4429)
      };
      \addplot[color=gray, style={very thick}] coordinates {
          (32, 19.8275)
          (32768, 19.8275)
      };

  \nextgroupplot[title={WebBase}]
      \addplot[mark options={solid,scale=0.5},mark=*,color=red, densely dashed] coordinates {
          (64, 10.9803)
          (128, 10.9893)
          (256, 12.5401)
          (512, 15.5037)
          (1024, 22.2099)
          (2048, 23.6108)
          (4096, 24.5998)
          (8192, 22.9578)
          (16384, 24.4941)
      };
      \addplot[mark options={solid,scale=0.5},mark=*,color=blue, densely dash dot dot] coordinates {
          (64, 9.70405)
          (128, 10.408)
          (256, 11.3147)
          (512, 13.8559)
          (1024, 18.3496)
          (2048, 22.8421)
          (4096, 21.5511)
          (8192, 21.4409)
          (16384, 21.5635)
      };
      \addplot[color=gray, style={very thick}] coordinates {
          (32, 19.8692)
          (32768, 19.8692)
      };

  \nextgroupplot[title={Friendster}]
      \addplot[mark options={solid,scale=0.5},mark=*,color=red, densely dashed] coordinates {
          (64, 42.6235)
          (128, 53.6999)
          (256, 64.4858)
          (512, 68.2374)
          (1024, 67.0643)
          (2048, 59.8266)
      };
      \addplot[mark options={solid,scale=0.5},mark=*,color=blue, densely dash dot dot] coordinates {
          (64, 36.0003)
          (128, 44.7288)
          (256, 53.2922)
          (512, 58.6392)
          (1024, 59.7824)
          (2048, 59.4441)
      };
      \addplot[color=gray, style={very thick}] coordinates {
          (32, 56.0829)
          (32768, 56.0829)
      };

  \nextgroupplot[title={blood vessel}]
      \addplot[mark options={solid,scale=0.5},mark=*,color=red, densely dashed] coordinates {
          (64, 0.739103)
          (128, 0.821112)
          (256, 0.950152)
          (512, 1.58013)
          (1024, 3.24364)
          (2048, 14.2545)
          (4096, 30.4395)
          (8192, 37.5675)
      };
      \addplot[color=gray, style={very thick}] coordinates {
          (32, 29.2982)
          (32768, 29.2982)
      };

  \nextgroupplot[title={cochlea}]
      \addplot[mark options={solid,scale=0.5},mark=*,color=red, densely dashed] coordinates {
          (64, 2.78413)
          (128, 3.06388)
          (256, 3.94344)
          (512, 4.78834)
          (1024, 9.31701)
          (2048, 20.967)
          (4096, 39.3943)
          (8192, 88.3876)
      };
      \addplot[color=gray, style={very thick}] coordinates {
          (32, 33.6489)
          (32768, 33.6489)
      };

  \end{groupplot}
  \node at ($(group c3r1)!.5!(group c4r1) + (0, 1.45cm)$) {\ref*{legend-approx-index-times}};
\end{tikzpicture}
\caption[Approximate index construction times]{Index construction times for
  {\ourscan} (48 cores with hyper-threading) using approximate similarity measures with varying
sample sizes.}
\Description{Plot of the times to construct the index with {\ourscan} using
  48 cores with hyper-threading with approximate similarity measures.
}
  \label{fig:approx-index-construction-times}
\end{figure*}

\subsubsection{Quality of approximate clusterings}\label{sec:experiment-quality}

The fourth experiment measures the quality of the clusterings achieved with the
approximate similarity measures compared to the clusterings achieved with the
exact similarity measures. Although the $\textsc{AssignNonCores}$
(\cref{alg:assign-non-cores}) portion of the clustering algorithm assigns each border
non-core vertex to the same cluster as an arbitrary $\ep$-similar core vertex,
to get consistent measurements for this experiment, we remove this
source of non-determinism by assigning each border vertex to the same
cluster as the most similar neighboring core vertex, breaking ties in favor of
lower vertex IDs.

We use the modularity as a heuristic measurement for clustering quality,
treating unclustered vertices as each being in its own cluster.
We select the parameters $(\mu, \ep)$ maximizing modularity from the
following set $\Sigma$:
\begin{equation}\label{eq:param-space}
  \Sigma = \set{2, 4, 8, 16, \ldots, 2^{18}} \times \set{.01, .02, .03, \ldots, .99}
.\end{equation}
\Cref{fig:time-vs-modularity} plots the best modularity scores found when
using the approximate similarity measures with varying numbers of samples. To better illustrate the trade-off between computation time and clustering
quality, we use the index construction times from
\cref{fig:approx-index-construction-times} on the horizontal axis instead of the
number of samples. Each plotted modularity score for a fixed number of samples
is the mean of five trials with different pseudorandom seeds on each trial.

Similarly, \cref{fig:time-vs-ari} plots the ARI of the clustering
found by the approximate similarity measures with varying numbers of samples
versus the ``ground-truth'' clustering under the corresponding exact similarity
measures. Again, each plotted ARI is the mean of five trials with different
pseudorandom seeds. The SCAN parameters used in this plot are the best
parameters in $\Sigma$ relative to the exact similarity measures.
Hence, this plot shows how well the clusterings using approximate
similarity measures match the clusterings using exact similarity
measures at a particular parameter setting.

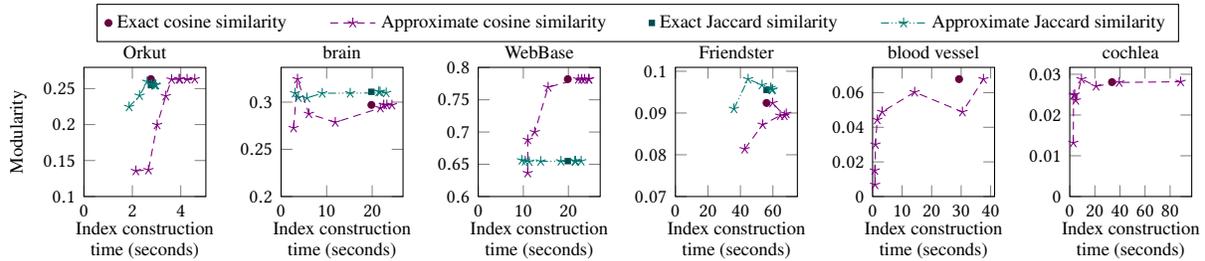
\begin{figure*} \centering
  \begin{tikzpicture}
  \begin{groupplot}[
    group style =
      {group size=6 by 1,
       xlabels at=edge bottom,
       ylabels at=edge left,
       horizontal sep=1cm,
     },
    height=0.15*\textheight,
    xlabel style={
      align=center, 
      at={(axis description cs:0.5,-0.13)} 
    },
    xlabel={Index construction\\ time (seconds)},
    ylabel={Modularity},
    xmin=0,
    scaled y ticks = false,
    yticklabel style={/pgf/number format/fixed},
    width=0.18*\textwidth,
    legend columns=-1, 
    legend style={font=\figfont, /tikz/every even column/.append style={column
      sep=0.5cm}}, 
    label style={font=\figfont},
    tick label style={font=\figfont},
    title style={font=\figfont, yshift=-6pt},
  ]
  \nextgroupplot[title={Orkut},
    ymin=0.1,
    legend to name=legend-time-vs-scores]

  \addplot[only marks,mark options={scale=0.65},mark=*,color=TyrianPurple] coordinates {
      (2.76594, 0.26318)
  };
  \addlegendentry[color=black]{Exact cosine similarity}
  \addplot[mark options={solid,scale=1},mark=star,color=violet, densely dashed] coordinates {
      (2.14565, 0.135675)
      (2.6652,  0.136876)
      (3.02249, 0.199628)
      (3.38223, 0.239714)
      (3.62116, 0.262724)
      (3.96687, 0.262935)
      (4.26101, 0.262992)
      (4.5783,  0.26318)
      (3.92902, 0.26318)
  };
  \addlegendentry[color=black]{Approximate cosine similarity}
  \addplot[only marks,mark options={scale=0.5},mark=square*,color=MidnightGreen] coordinates {
      (2.76594, 0.255042)
  };
  \addlegendentry[color=black]{Exact Jaccard similarity}
  \addplot[mark options={solid,scale=1},mark=star,color=PineGreen, densely dash dot dot] coordinates {
      (1.86581, 0.224617)
      (2.29971, 0.240287)
      (2.61267, 0.259563)
      (2.80486, 0.256001)
      (2.94262, 0.255098)
      (2.97076, 0.255042)
      (2.94041, 0.255044)
      (2.92462, 0.255043)
      (2.92243, 0.255042)
  };
  \addlegendentry[color=black]{Approximate Jaccard similarity}

  \nextgroupplot[title={brain},
    ymin=0.2,
  ]
  \addplot[only marks,mark options={scale=0.65},mark=*,color=TyrianPurple] coordinates {
      (19.8275, 0.297052)
  };
  \addplot[mark options={solid,scale=1},mark=star,color=violet, densely dashed] coordinates {
      (2.74868, 0.27265)
      (3.60887, 0.324388)
      (6.10136, 0.287662)
      (11.899,  0.278608)
      (21.8003, 0.293787)
      (24.3374, 0.296942)
      (23.2177, 0.297052)
      (22.5325, 0.297052)
  };
  \addplot[only marks,mark options={scale=0.5},mark=square*,color=MidnightGreen] coordinates {
      (19.8275, 0.310942)
  };
  \addplot[mark options={solid,scale=1},mark=star,color=PineGreen, densely dash dot dot] coordinates {
      (3.05867, 0.309546)
      (3.79266, 0.304903)
      (5.67183, 0.304448)
      (9.03514, 0.309428)
      (15.1697, 0.309509)
      (23.0833, 0.309859)
      (21.729,  0.310957)
      (21.4429, 0.310942)
  };

  \nextgroupplot[title={WebBase},
    ymin=0.6,
  ]
  \addplot[only marks,mark options={scale=0.65},mark=*,color=TyrianPurple] coordinates {
      (19.8692, 0.782014)
  };

  \addplot[mark options={solid,scale=1},mark=star,color=violet, densely dashed] coordinates {
    (10.9803, 0.636527)
    (10.9893, 0.687991)
    (12.5401, 0.699998)
    (15.5037, 0.769551)
    (22.2099, 0.781813)
    (23.6108, 0.781835)
    (24.5998, 0.781723)
    (22.9578, 0.781853)
    (24.4941, 0.782072)
  };
  \addplot[only marks,mark options={scale=0.5},mark=square*,color=MidnightGreen] coordinates {
      (19.8692, 0.654896)
  };
  \addplot[mark options={solid,scale=1},mark=star,color=PineGreen, densely dash dot dot] coordinates {
    (9.70405,0.656228)
    (10.408, 0.655009)
    (11.3147,0.654175)
    (13.8559,0.654537)
    (18.3496,0.654786)
    (22.8421,0.654915)
    (21.5511,0.654897)
    (21.4409,0.654894)
    (21.5635,0.654896)
  };

  \nextgroupplot[title={Friendster},
    ymin=0.07,
  ]
  \addplot[only marks,mark options={scale=0.65},mark=*,color=TyrianPurple] coordinates {
      (56.0829, 0.0924098)
  };
  \addplot[mark options={solid,scale=1},mark=star,color=violet, densely dashed] coordinates {
      (42.6235, 0.0813536)
      (53.6999, 0.0872217)
      (64.4858, 0.0893867)
      (68.2374, 0.0898194)
      (67.0643, 0.0893517)
      (59.8266, 0.0924098)
  };
  \addplot[only marks,mark options={scale=0.5},mark=square*,color=MidnightGreen] coordinates {
      (56.0829, 0.0955208)
  };
  \addplot[mark options={solid,scale=1},mark=star,color=PineGreen, densely dash dot dot] coordinates {
      (36.0003,0.0910133)
      (44.7288,0.0981195)
      (53.2922,0.0967084)
      (58.6392,0.0960252)
      (59.7824,0.0956515)
      (59.4441,0.0955319)
  };

  \nextgroupplot[title={blood vessel},
    ymin=0,
  ]
  \addplot[only marks,mark options={scale=0.65},mark=*,color=TyrianPurple] coordinates {
      (29.2982, 0.0678182)
  };
  \addplot[mark options={solid,scale=1},mark=star,color=violet, densely dashed] coordinates {
      (0.739103,0.0150673)
      (0.821112,0.00682346)
      (0.950152,0.0300468)
      (1.58013, 0.0443178)
      (3.24364, 0.0489213)
      (14.2545, 0.0603051)
      (30.4395, 0.0488652)
      (37.5675, 0.0678182)
  };

  \nextgroupplot[title={cochlea},
    ymin=0,
  ]
  \addplot[only marks,mark options={scale=0.65},mark=*,color=TyrianPurple] coordinates {
      (33.6489, 0.0280733)
  };
  \addplot[mark options={solid,scale=1},mark=star,color=violet, densely dashed] coordinates {
      (2.78413, 0.0131848)
      (3.06388, 0.0248524)
      (3.94344, 0.0250288)
      (4.78834, 0.0236121)
      (9.31701, 0.0288123)
      (20.967,  0.0270835)
      (39.3943, 0.0279982)
      (88.3876, 0.0281563)
  };

  \end{groupplot}
  \node at ($(group c3r1)!.5!(group c4r1) + (0, 1.45cm)$) {\ref*{legend-time-vs-scores}};
  \end{tikzpicture}
  \caption[Approximate index construction time versus modularity score]{
    Trade-off curve of approximate similarity index construction time
    with varying numbers of samples versus the best modularity score found among
    parameters from $\Sigma$ (\cref{eq:param-space}). The index construction
    times on the horizontal axis come from
    \cref{fig:approx-index-construction-times}.
  }
  \Description{
    Plot comparing the time to construct the SCAN index with approximate
    similarity measures of varying numbers of samples against the best
    modularity score found using that index.
  }
  \label{fig:time-vs-modularity}
\end{figure*}

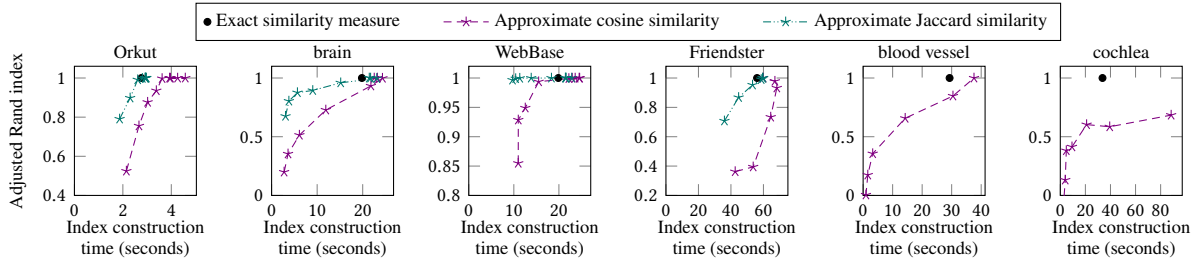
\begin{figure*} \centering
  \begin{tikzpicture}
  \begin{groupplot}[
    group style =
      {group size=6 by 1,
       xlabels at=edge bottom,
       ylabels at=edge left,
       horizontal sep=1cm,
     },
    height=0.15*\textheight,
    xlabel style={
      align=center, 
      at={(axis description cs:0.5,-0.13)} 
    },
    xlabel={Index construction\\ time (seconds)},
    ylabel={Adjusted Rand index},
    xmin=0,
    width=0.18*\textwidth,
    legend columns=-1, 
    legend style={font=\figfont, /tikz/every even column/.append style={column sep=0.5cm}}, 
    label style={font=\figfont},
    tick label style={font=\figfont},
    title style={font=\figfont, yshift=-6pt},
  ]
  \nextgroupplot[title={Orkut},
    ymin=0.4,
    legend to name=legend-time-vs-ari]

  \addplot[only marks,mark options={scale=0.65},mark=*,color=black] coordinates {
      (2.76594, 1)
  };
  \addlegendentry[color=black]{Exact similarity measure}
  \addplot[mark options={solid,scale=1},mark=star,color=violet, densely dashed] coordinates {
      (2.14565, 0.524468)
      (2.6652,  0.755397)
      (3.02249, 0.876263)
      (3.38223, 0.935462)
      (3.62116, 0.997363)
      (3.96687, 0.999456)
      (4.26101, 0.999776)
      (4.5783,  1)
      (3.92902, 1)
  };
  \addlegendentry[color=black]{Approximate cosine similarity}
  \addplot[mark options={solid,scale=1},mark=star,color=PineGreen, densely dash dot dot] coordinates {
      (1.86581, 0.790528)
      (2.29971, 0.897974)
      (2.61267, 0.991192)
      (2.80486, 0.998294)
      (2.94262, 0.999907)
      (2.97076, 0.999989)
      (2.94041, 0.999998)
      (2.92462, 1)
      (2.92243, 1)
  };
  \addlegendentry[color=black]{Approximate Jaccard similarity}

  \nextgroupplot[title={brain},
    ymin=0,
  ]
  \addplot[only marks,mark options={scale=0.65},mark=*,color=black] coordinates {
      (19.8275, 1)
  };
  \addplot[mark options={solid,scale=1},mark=star,color=violet, densely dashed] coordinates {
      (2.74868, 0.199201)
      (3.60887, 0.353546)
      (6.10136, 0.513414)
      (11.899,  0.727621)
      (21.8003, 0.932531)
      (24.3374, 0.998401)
      (23.2177, 1)
      (22.5325, 1)
  };
  \addplot[mark options={solid,scale=1},mark=star,color=PineGreen, densely dash dot dot] coordinates {
      (3.05867, 0.6746)
      (3.79266, 0.803467)
      (5.67183, 0.875522)
      (9.03514, 0.894073)
      (15.1697, 0.959281)
      (23.0833, 0.996383)
      (21.729,  0.99999)
      (21.4429, 1)
  };

\nextgroupplot[title={WebBase},
    ymin=0.8,
]
  \addplot[only marks,mark options={scale=0.65},mark=*,color=black] coordinates {
      (19.8692, 1)
  };

  \addplot[mark options={solid,scale=1},mark=star,color=violet, densely dashed] coordinates {
      (10.9803,0.854845)
      (10.9893,0.928523)
      (12.5401,0.948991)
      (15.5037,0.993648)
      (22.2099,0.999155)
      (23.6108,0.999904)
      (24.5998,0.999759)
      (22.9578,0.999802)
      (24.4941,0.999965)
  };
  \addplot[mark options={solid,scale=1},mark=star,color=PineGreen, densely dash dot dot] coordinates {
      (9.70405, 0.996501)
      (10.408,  0.99895)
      (11.3147, .999641)
      (13.8559, .999891)
      (18.3496, .999948)
      (22.8421, .999973)
      (21.5511, 1)
      (21.4409, 1)
      (21.5635, 1)
  };

  \nextgroupplot[title={Friendster},
    ymin=0.2,
  ]
  \addplot[only marks,mark options={scale=0.65},mark=*,color=black] coordinates {
      (56.0829, 1)
  };

  \addplot[mark options={solid,scale=1},mark=star,color=violet, densely dashed] coordinates {
      (42.6235, .362255)
      (53.6999, .394418)
      (64.4858, .733761)
      (68.2374, .932485)
      (67.0643, .97816)
      (59.8266, 1)
  };
  \addplot[mark options={solid,scale=1},mark=star,color=PineGreen, densely dash dot dot] coordinates {
      (36.0003,.707586)
      (44.7288,.867583)
      (53.2922,.952617)
      (58.6392,.989848)
      (59.7824,.999119)
      (59.4441,.999948)
  };

  \nextgroupplot[title={blood vessel},
    ymin=0,
  ]
  \addplot[only marks,mark options={scale=0.65},mark=*,color=black] coordinates {
      (29.2982, 1)
  };
  \addplot[mark options={solid,scale=1},mark=star,color=violet, densely dashed] coordinates {
      (0.739103, -0.00102727)
      (0.821112,-0.000408073)
      (0.950152,0.00413669)
      (1.58013, 0.173818)
      (3.24364, 0.35588)
      (14.2545, 0.657302)
      (30.4395, 0.846927)
      (37.5675, 1)
  };

  \nextgroupplot[title={cochlea},
    ymin=0,
  ]
  \addplot[only marks,mark options={scale=0.65},mark=*,color=black] coordinates {
      (33.6489, 1)
  };
  \addplot[mark options={solid,scale=1},mark=star,color=violet, densely dashed] coordinates {
      (2.78413,-0.00546551)
      (3.06388,-0.000734335)
      (3.94344,0.13127)
      (4.78834,0.381686)
      (9.31701,0.41377)
      (20.967, 0.603283)
      (39.3943,0.58483)
      (88.3876,0.683936)
  };

  \end{groupplot}
  \node at ($(group c3r1)!.5!(group c4r1) + (0, 1.45cm)$) {\ref*{legend-time-vs-ari}};
  \end{tikzpicture}
  \caption[Approximate index construction time versus clustering accuracy]{
    Trade-off curve of approximate similarity index construction time
    with varying numbers of samples versus the accuracy (measured via adjusted
    Rand index) of the resulting approximate clustering relative to a
    ``ground-truth'' clustering from the corresponding exact similarity index.
    For each graph, the tested parameters $(\mu, \ep)$ are the
    modularity-maximizing parameters from $\Sigma$ (\cref{eq:param-space})
    relative to the exact similarity measure.  The index construction times on
    the horizontal axis come from \cref{fig:approx-index-construction-times}.
  }
  \Description{
    Plot comparing the time to construct the SCAN index with approximate
    similarity measures of varying numbers of samples against the
    adjusted Rand index achieved using that index.
  }
  \label{fig:time-vs-ari}
\end{figure*}

Points to the top and the left represent sample sizes that give good quality
as well as low index construction times. The figures also include the times
to construct indices with exact similarity measures from
\cref{fig:index-construction-times}, with the assumption that the times
for exact Jaccard similarity are the same as those measured for
cosine similarity.

The improved approximation accuracy in these plots as the sample size increases is
not only attributable to better LSH accuracy with
more samples, but also to the heuristic described in
\cref{sec:approx-sims-impl} that reverts to computing exact similarity for
vertices that have low degree relative to the number of samples.

The approximate Jaccard clusterings approach the quality of the
corresponding exact similarity clusterings at lower sample sizes than
approximate cosine clusterings do, which is perhaps expected due to
the better sampling efficiency that MinHash variants tend to have over
SimHash, as suggested by Shrivastava and Li~\cite{shrivastava2014defense} and by
our bounds in
\cref{thm:simhash-accuracy-bound,thm:minhash-accuracy-bound}.

The modularity and ARI scores indicate that approximating similarities
can significantly speed up index construction while still achieving good quality
clusterings. The modularity plots in \cref{fig:time-vs-modularity} look
more favorable than the ARI plots in \cref{fig:time-vs-ari}, suggesting that
though at low sample sizes the approximate clusters at a particular parameter
setting may noticeably differ from the corresponding exact clusters, we are
still able to find a good quality clustering by searching over a range of
parameter values.

\section{Related Work}

Xu et al.\ introduced the original SCAN algorithm~\cite{xu2007scan} and borrowed ideas
from the popular spatial clustering algorithm DBSCAN~\cite{ester1996density}.
One major inconvenience of SCAN is the difficulty of choosing its two user-selected
parameters, $\mu$ and $\ep$.
GS*-Index alleviates this issue by creating an index upon which future SCAN
queries with arbitrary parameters are efficient~\cite{wen2017efficient}.
SCOT~\cite{bortner2010progressive} and gSkeletonClu~\cite{huang2012revealing}
also essentially compute indices for SCAN, but only for a fixed $\mu$ value. SCOT outputs
an ordering of vertices, similar to what the OPTICS
algorithm~\cite{ankerst1999optics} outputs for DBSCAN, such that vertices that
tend to be in the same cluster are nearby in the ordering. gSkeletonClu
computes a spanning tree on potential core vertices.

SHRINK~\cite{huang2010shrink}, DHSCAN~\cite{yuruk2007divisive}, and
AHSCAN~\cite{yuruk2009ahscan} are all based on SCAN, but avoid the parameter
selection issue by being parameter-free algorithms that use a quality function
like the modularity to guide the clustering process.
DPSCAN~\cite{wu2019dpscan} is another parameter-free SCAN-based algorithm that
uses a density metric to select clusters. These algorithms are easier to use
due to their lack of parameters, although having tunable
parameters can be helpful in allowing the user to explore alternative
clusterings.

Other work building on SCAN focuses on making SCAN scale to large graphs.
LinkSCAN*~\cite{lim2014linkscan} reduces computation time at the cost of
accuracy by operating on a sampled subgraph of the original graph. It may be
worthwhile in the future to compare the efficiency and clustering quality of the
LinkSCAN* sampling approach versus the LSH approach of
our paper.
Zhao et al.~\cite{zhao2017anyscan} and Mai et al.~\cite{mai2018scalable}
describe anytime algorithms for SCAN, with Mai et al.'s algorithm being
parallel. Users may pause queries and examine intermediate clustering results,
making it useful for large graphs on which finishing a query may take a long
time. Our work, on the other hand, strives to make finishing a query as fast as
possible so that this anytime functionality is unnecessary.

SCAN++~\cite{shiokawa2015scan++}, pSCAN~\cite{chang2017pscan}, and
ppSCAN~\cite{che2018parallelizing}, for a fixed
setting of SCAN parameters, speed up SCAN by pruning many unnecessary similarity score computations
between pairs of vertices.
Che et al.'s ppSCAN is parallel and uses vectorized instructions as well for additional
performance.
SCAN-XP~\cite{takahashi2017scan} is another parallel SCAN algorithm but does not
perform pruning.

For distributed systems, Chen et al.~\cite{chen2013pscan} and
Zhao et al.~\cite{zhao2013pscan} present MapReduce
parallelizations of SCAN, and SparkSCAN~\cite{zhou2015sparkscan} is a Spark
parallelization of SCAN.
GPUSCAN~\cite{stovall2014gpuscan}  uses GPUs to
speed up SCAN.

There are many other graph clustering algorithms besides SCAN and its variants.
Interested readers may refer to surveys written by other researchers, such as
Schaeffer~\cite{schaeffer2007graph} and Fortunato~\cite{fortunato2010community}.

\section{Conclusion}

This paper presents index-based SCAN algorithms that achieve significant
parallel speedup over the state of the art. They allow users to query efficiently for SCAN clusterings for
arbitrary parameter settings. The algorithms achieve improved work bounds over GS*-Index and have
logarithmic span \whp.
We also present an optimized
multicore implementation of the algorithm that runs well in practice. Moreover,
we demonstrate that LSH is a viable approximation scheme
to speed up the computationally expensive component of index construction.

For future work, first, we are interested in extending our work
to dynamic graphs by devising parallel algorithms for
processing batches of
edge updates. Second, we are
interested in quickly extracting hierarchical clusterings from the SCAN index.
Third, we would like to investigate the speed and clustering quality of SCAN when using other similarity measures.
Last, we wish to compare SCAN to other parallel
clustering algorithms in quality and speed.


\begin{acks}
We thank Rezaul Chowdhury for suggesting using matrix multiplication on dense graphs.
  This research was supported by DOE Early Career Award \#DE-SC0018947, NSF
CAREER Award \#CCF-1845763, Google Faculty Research Award, DARPA SDH Award \#HR0011-18-3-0007, and
Applications Driving Architectures (ADA) Research Center, a JUMP
Center co-sponsored by SRC and DARPA.
\end{acks}

\bibliographystyle{ACM-Reference-Format}
\bibliography{references}

\end{document}